\documentclass[aoas]{imsart}

\RequirePackage{amsthm,amsmath,amsfonts,amssymb}
\RequirePackage[authoryear]{natbib}
\RequirePackage[colorlinks,citecolor=blue,urlcolor=blue]{hyperref}
\RequirePackage{graphicx}
\usepackage{float}

\usepackage[ruled]{algorithm2e}
\SetKwInput{KwParam}{Parameter}
\SetAlgoCaptionLayout{centerline}

\usepackage{tikz}
\usetikzlibrary{shapes,decorations,arrows,calc,arrows.meta,fit,positioning}
\usepackage{blkarray}
\usepackage{bbm}
\usepackage{multicol}
\usepackage{subfig}
\usepackage{multirow}

\usepackage{mathtools}

\usepackage{xcolor}

\startlocaldefs
\theoremstyle{plain}

\newtheorem{theorem}{Theorem}[section]
\newtheorem{lemma}[theorem]{Lemma}
\newtheorem{corollary}[theorem]{Corollary}
\newtheorem{proposition}[theorem]{Proposition}
\theoremstyle{remark}
\newtheorem{definition}[theorem]{Definition}

\newtheorem*{remark}{Remark}

\DeclareMathOperator*{\argmin}{arg\,min}

\makeatletter
\newcommand*\bigcdot{\mathpalette\bigcdot@{.5}}
\newcommand*\bigcdot@[2]{\mathbin{\vcenter{\hbox{\scalebox{#2}{$\m@th#1\bullet$}}}}}
\makeatother

\endlocaldefs

\begin{document}

\begin{frontmatter}
\title{Statistical matching and subclassification with a continuous dose: characterization, algorithm, and application to a health outcomes study}
\runtitle{Statistical matching and subclassification with a continuous dose}

\begin{aug}
\author[A]{\fnms{Bo} \snm{Zhang}\ead[label=e1]{bozhan@wharton.upenn.edu}},
\author[B]{\fnms{Emily J.} \snm{Mackay}\ead[label=e2,mark]{Emily.Mackay@pennmedicine.upenn.edu}}
\and
\author[C]{\fnms{Mike} \snm{Baiocchi}\ead[label=e3,mark]{baiocchi@stanford.edu}}
\address[A]{Department of Statistics and Data Science, The Wharton School, University of Pennsylvania,
\printead{e1}}

\address[B]{Department of Anesthesiology and Critical Care, Perelman School of Medicine, University of Pennsylvania,
\printead{e2}}

\address[C]{Department of Epidemiology and Population Health, Stanford University,
\printead{e3}}
\end{aug}

\begin{abstract}
Subclassification and matching are often used in empirical studies to adjust for observed covariates; however, they are largely restricted to relatively simple study designs with a binary treatment and less developed for designs with a continuous exposure. Matching with exposure doses is particularly useful in instrumental variable designs and in understanding the dose-response relationships. In this article, we propose two criteria for optimal subclassification based on subclass homogeneity in the context of having a continuous exposure dose, and propose an efficient polynomial-time algorithm that is guaranteed to find an optimal subclassification with respect to one criterion and serves as a 2-approximation algorithm for the other criterion. We discuss how to incorporate dose and use appropriate penalties to control the number of subclasses in the design. Via extensive simulations, we systematically compare our proposed design to optimal non-bipartite pair matching, and demonstrate that combining our proposed subclassification scheme with regression adjustment helps reduce model dependence for parametric causal inference with a continuous dose. We apply the new design and associated randomization-based inferential procedure to study the effect of transesophageal echocardiography (TEE) monitoring during coronary artery bypass graft (CABG) surgery on patients' post-surgery clinical outcomes using Medicare and Medicaid claims data, and find evidence that TEE monitoring lowers patients' all-cause $30$-day mortality rate.

\end{abstract}


\end{frontmatter}

\section{Introduction}
\subsection{Application: The effect of TEE monitoring during CABG surgery}
Transesophageal echocardiography (henceforth TEE) is an ultrasound-based, cardiac imaging modality often used in cardiac surgeries to monitor patients' hemodynamics. TEE may potentially improve post-surgery clinical outcomes by facilitating intraoperative surgery decision making and managing complications related to cardiopulmonary bypass (\citealp{hahn2013guidelines,nishimura20172017}); indeed, \citet{mackay2020transesophageal} found perioperative TEE use was associated with lower $30$-day all-cause mortality among patients undergoing open cardiac valve repair or replacement surgery. Coronary artery bypass graft (henceforth CABG) surgery is the most widely performed surgery in the United States (\citealp{Database_2016}). Compared to open valve surgery, evidence supporting the use of TEE during isolated CABG surgery is more equivocal: TEE monitoring is classified by American Heart Association/American College of Cardiology (AHA/ACC) guidelines as a Class IIb recommendation, meaning its ``usefulness/efficacy is less well established by evidence/opinion" (\citealp{hillis20112011}). \citet{MacKay2020_protocol,mackay2021association} proposed to study TEE's effect on clinical outcomes using providers' preference for TEE as an instrumental variable (IV). 

One challenge in the study design is that the IV-defined exposure, providers' preference in this case, is continuous. A straightforward strategy to deal with a continuous exposure is to dichotomize it according to some pre-specified dichotomization scheme. Despite its simplicity and popularity, this practice suffers from at least two major drawbacks. First, defining the potential outcome under a dichotomized exposure potentially violates the stable unit treatment value assumption (SUTVA) (\citealp{rubin1980randomization, rubin1986statistics}). Let $\widetilde{Z}$ denote the continuous exposure and $Z = \mathbbm{1}\{\widetilde{Z} > c^\ast\}$ the associated dichotomized version, e.g., $c^\ast$ being the median. The potential outcome under $Z = 1$ is well-defined only when the potential outcome remains unchanged for all exposure doses $\widetilde{Z}$ exceeding the pre-specified threshold $c^\ast$. This is at best an approximation to the complicated reality in most circumstances. Moreover, dichotomizing the continuous exposure inevitably censors the rich information contained in the original dose and prevents researchers from investigating a dose-response relationship. Therefore, we would prefer a study design that preserves the continuous exposure dose (\citealt{lopez2017estimation}).

In their original study protocol, \citet{MacKay2020_protocol} embed observational data from Centers for Medicare and Medicaid Services (CMS) into a paired cluster-randomized encouragement experiment. \citet{MacKay2020_protocol} matched hospitals with similar patient population and hospital-level characteristics but distinct preference for TEE using a design technique called \emph{optimal non-bipartite pair match} (\citealp{lu2001matching, lu2011optimal, baiocchi2010building, baiocchi2012near}). A non-bipartite matching algorithm is distinct from bipartite matching algorithms suited only for statistical matching and subclassification with a binary treatment (\citealp{rosenbaum1989optimal, rosenbaum1991characterization, stuart2010matching}) and matching algorithms based on generalized propensity score (\citealp{yang2016propensity, lopez2017estimation, wu2018matching}); see Supplementary Material A for a detailed literature review. Figure \ref{fig: bipartite vs non-bipartite} gives a graphical representation of pair matching in a bipartite and a non-bipartite setting. With a binary treatment, there are well-defined treated and control groups, and there is little ambiguity in the ultimate goal of statistical matching: the matching algorithm aims to ``construct" or ``design" a matched control group (or comparison group) that resembles the treated group in baseline covariates. On the other hand, with a continuous exposure, there are no pre-defined treated and control groups: in principle, any unit can be matched to any other unit similar in covariates. This structural difference between bipartite and non-bipartite settings makes it more challenging to characterize optimal subclassification and design efficient algorithms in the non-bipartite context.

\begin{figure}[ht]
   \begin{tikzpicture}
\draw[help lines, color=gray!30, dashed] (0,1) grid (4.9,4.9);
\draw[->,ultra thick] (0,1)--(5,1) node[right]{};
\draw[->,ultra thick] (0,1)--(0,5) node[above]{\large Exposure};


    \filldraw [black] (0.7,4) circle (2pt);
    \filldraw [black] (2,4) circle (2pt);
    \filldraw [black] (4.2,4) circle (2pt);

    \filldraw [black] (0.25,2) circle (2pt);
    \filldraw [black] (1,2) circle (2pt);
    \filldraw [black] (2.2,2) circle (2pt);
    \filldraw [black] (1.2,2) circle (2pt);
    \filldraw [black] (2.7,2) circle (2pt);
    \filldraw [black] (3.7,2) circle (2pt);

    \draw [ultra thick] (0.7, 4) -- (1, 2); 
    \draw [ultra thick] (2, 4) -- (2.2, 2); 
    \draw [ultra thick] (4.2, 4) -- (3.7, 2); 

    \node [text width = 0.5 cm, text centered] at (-1.5, 4) {\large Treated};
    \node [text width = 0.5 cm, text centered] at (-1.5, 2) {\large Control};
    \node [text width = 0.5 cm, text centered] at (2.5, 0.5) {\Large X};

    \draw[help lines, color=gray!30, dashed] (7,1) grid (11.9,4.9);
    \draw[->,ultra thick] (7,1)--(12,1) node[right]{};
    \draw[->,ultra thick] (7,1)--(7,5) node[above]{\large Exposure};

    \filldraw [black] (7.5,4.5) circle (2pt);
    \filldraw [black] (8.7,3.8) circle (2pt);

    \filldraw [black] (7.7,1.5) circle (2pt);
    \filldraw [black] (8.2,2.5) circle (2pt);
    
    \filldraw [black] (10.5,1.2) circle (2pt);
    \filldraw [black] (11.2,1.7) circle (2pt);
    
    \filldraw [black] (11.5,4) circle (2pt);
    \filldraw [black] (9.7,2.3) circle (2pt);

    \draw [ultra thick] (7.5, 4.5) -- (8.7, 3.8); 
    \draw [ultra thick] (7.7, 1.5) -- (8.2, 2.5); 
    \draw [ultra thick] (10.5, 1.2) -- (11.2, 1.7); 
    \draw [ultra thick] (11.5, 4) -- (9.7, 2.3); 
    
    \node [text width = 0.5 cm, text centered] at (6, 3) {\large Dose};
    \node [text width = 0.5 cm, text centered] at (9.5, 0.5) {\Large X};

\end{tikzpicture}
 \caption{\small Left panel: a pair match in a bipartite setting with a binary treatment. Right panel: a pair match in a non-bipartite setting with a many-level or continuous treatment dose.}
    \label{fig: bipartite vs non-bipartite}
\end{figure}
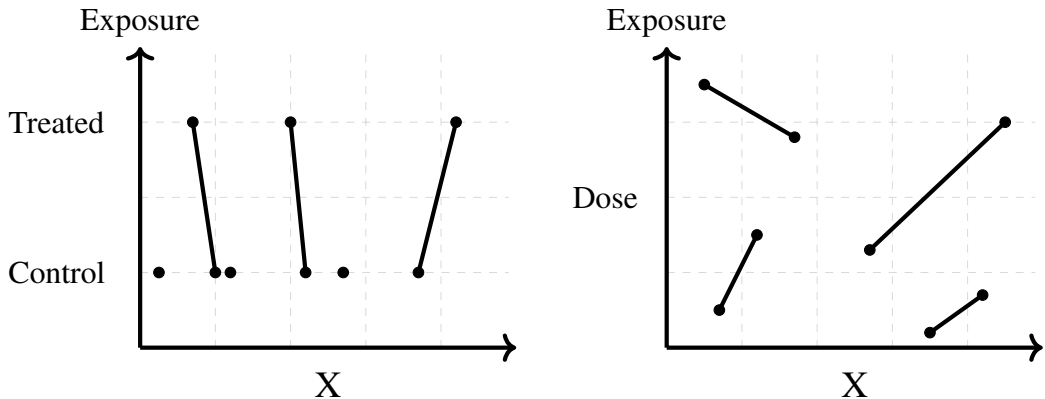

\subsection{Pair matching is not optimal}
Full matching is a more flexible subclassification scheme that divides units into non-overlapping matched sets (or subclasses) of size at least, but not necessarily, equal to two. With a binary treatment, \citet{rosenbaum1991characterization} found:
\begin{quote}
    [T]here may be no pair matching and no matching with multiple controls that is an optimal subclassification. ... [A] best pair match may be arbitrarily poor compared with the optimal full matching.
\end{quote}
These statements remain true in the non-bipartite setting with a continuous dose. To see this, it suffices to consider the following simple example with $6$ units $\{a, b, c, d, e, f\}$ and the associated distance matrix
\[\textbf{M} ~~=~~~
\begin{blockarray}{ccccccc}
& a & b & c & d & e & f\\
\begin{block}{c(cccccc)}
  a & 0 & \epsilon & \epsilon & \omega & \omega & \omega \\
  b & \epsilon & 0 & \epsilon & \omega & \omega & \omega \\
  c & \epsilon & \epsilon & 0 & \omega & \omega & \omega \\
  d & \omega & \omega & \omega & 0 & \epsilon & \epsilon  \\
  e & \omega & \omega & \omega & \epsilon & 0 &\epsilon  \\
  f & \omega & \omega & \omega & \epsilon & \epsilon &0 \\
\end{block}
\end{blockarray}
 \]
with $\epsilon \ll \omega$. The $ij^{\text{th}}$ entry of \textbf{M} represents a measure of distance, e.g., the Mahalanobis distance of observed covariates, between unit $i$ and $j$. An optimal pair match produces the following three matched pairs: \[
\boldsymbol\Pi_{\text{pair}} = \left\{\{a, c\}, \{b, d\}, \{e, f\} \right\}.
\]
On the other hand, consider the following full match:
\[
\boldsymbol\Pi_{\text{full}} = \left\{\{a, b, c\}, \{d, e, f\}\right\}. 
\]
It is evident that $\boldsymbol\Pi_{\text{full}}$ achieves a better matched-sets homogeneity, which we will carefully define later, compared to $\boldsymbol\Pi_{\text{pair}}$ when $\epsilon \ll \omega$; see Figure \ref{fig: pair match vs full match} for a transparent graphical representation. Moreover, since $\omega$ can be arbitrarily larger than $\epsilon$, $\boldsymbol\Pi_{\text{full}}$ can be arbitrarily better than $\boldsymbol\Pi_{\text{pair}}$ according to any reasonable homogeneity measure. In the most extreme case where $\omega = \infty$, there exists no admissible pair match exhausting six units; however, there exists a feasible full match, and a good one when $\epsilon$ is small.

\begin{figure}[h]
   \centering
     \subfloat{\begin{tikzpicture}[x = 1cm, y = 0.5cm]
   \foreach \Point in {(1, 1), (2,3), (3, 1), (4.5, 1), (5.5, 3), (6.5, 1)}{
    \node at \Point {\textbullet};
    }
    
    \draw [ultra thick] (2, 3) -- (1, 1); 
    \draw [ultra thick] (2, 3) -- (3, 1); 
    \draw [ultra thick] (5.5, 3) -- (4.5, 1);
    \draw [ultra thick] (5.5, 3) -- (6.5, 1);
    \node [text width = 0.5 cm] at (1.5, 2.5) {$\epsilon$};
    \node [text width = 0.5 cm] at (2.85, 2.5) {$\epsilon$};
    \node [text width = 0.5 cm] at (5, 2.5) {$\epsilon$};
    \node [text width = 0.5 cm] at (6.35, 2.5) {$\epsilon$};
    \node [text width = 0.5 cm] at (1.15, 0.5) {$a$};
    \node [text width = 0.5 cm] at (3.15, 0.5) {$b$};
    \node [text width = 0.5 cm] at (2.15, 3.6) {$c$};
    \node [text width = 0.5 cm] at (4.65, 0.5) {$d$};
    \node [text width = 0.5 cm] at (6.65, 0.5) {$e$};
    \node [text width = 0.5 cm] at (5.65, 3.6) {$f$};
    \end{tikzpicture}}\hspace{1.5cm}
  \subfloat{\begin{tikzpicture}[x = 1cm, y = 0.5cm]
   \foreach \Point in {(1, 1), (2,3), (3, 1), (4.5, 1), (5.5, 3), (6.5, 1)}{
    \node at \Point {\textbullet};
    }
    
    \draw [ultra thick] (2, 3) -- (1, 1); 
    \draw [ultra thick] (5.5, 3) -- (6.5, 1);
    \draw [ultra thick] (3, 1) -- (4.5, 1);
    \node [text width = 0.5 cm] at (1.5, 2.5) {$\epsilon$};
    \node [text width = 0.5 cm] at (6.35, 2.5) {$\epsilon$};
    \node [text width = 0.5 cm] at (3.85, 1.5) {$\omega$};
    \node [text width = 0.5 cm] at (1.15, 0.5) {$a$};
    \node [text width = 0.5 cm] at (3.15, 0.5) {$b$};
    \node [text width = 0.5 cm] at (2.15, 3.6) {$c$};
    \node [text width = 0.5 cm] at (4.65, 0.5) {$d$};
    \node [text width = 0.5 cm] at (6.65, 0.5) {$e$};
    \node [text width = 0.5 cm] at (5.65, 3.6) {$f$};
    \end{tikzpicture}}
    \caption{\small Left panel: a full match $\boldsymbol\Pi_{\text{full}}$ dividing six units into two matched sets of size three each. Right panel: a pair match $\boldsymbol\Pi_{\text{pair}}$ dividing six units into three matched pairs. When $\omega \gg \epsilon$, $\boldsymbol\Pi_{\text{full}}$ delivers matched sets with superior homogeneity than $\boldsymbol\Pi_{\text{pair}}$.}
    \label{fig: pair match vs full match}
\end{figure}
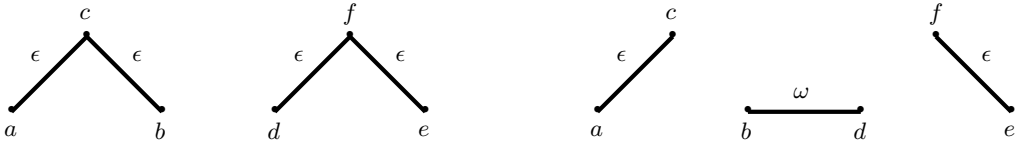

Another concern with pair matching is that it is often not flexible enough to deliver a subclassification that is simultaneously homogeneous in units' observed covariates and reasonably heterogeneous in units' exposure doses; in fact, units are often removed in the design stage (\citealp{baiocchi2010building, baiocchi2012near}) to achieve both goals. For instance, \citet{MacKay2020_protocol} removed $20\%$ of all hospitals in their matched-pair design using a design device known as ``sinks" (\citealp{baiocchi2010building}). Ideally, we would prefer a design that utilizes all units while maintaining homogeneity in covariates and good separation in exposure doses. Lastly, one minor issue with pair matching is that, when the number of units is odd, say $N = 5$, the design necessarily discards one unit to produce two matched pairs.

These limitations of a non-bipartite pair match design and the abundance of observational studies with a continuous or many-level exposure motivate us to study optimal subclassfication in the non-bipartite setting.

\subsection{Outline: a characterization of optimal non-bipartite subclassification, an algorithm, two simulation studies, and an application}
Two subclassification homogeneity measures and optimal subclassification with respect to each measure are defined in Section \ref{sec: two homo measures}. Section \ref{sec: relationship between two solutions} proves a useful relationship between the two homogeneity measures; this relationship suggests that any algorithm that finds a subclassification with respect to one homogeneity measure is automatically a 2-approximation algorithm for the other measure. An efficient, polynomial-time algorithm that finds an optimal subclassification with respect to one homogeneity measure and suitable weights is presented in Section \ref{sec: algorithm exist for Pi^astast}. Section \ref{sec: other design aspects} discusses how to further incorporate the treatment dose in the design stage, and how to probe the middle ground between an optimal pair match and an optimal subclassification. Two simulation studies, one examining how combining the proposed subclassification scheme with regression adjustment helps reduce bias of the regression estimator, and the other systematically comparing the proposed subclassification method to optimal pair matching, are presented in Section \ref{sec: nbp matching as preprocessing} and \ref{sec: nbp matching full vs pair}, respectively. We leverage the proposed novel design and conduct randomized-based inference to study the effect of TEE monitoring during CABG surgery on patients' 30-day all-cause mortality in Section \ref{sec: real data study design}. We conclude with a brief discussion in Section \ref{sec: discussion}.

\section{Two measures of subclassification homogeneity}
\label{sec: two homo measures}
Let $\mathcal{N} = \{1, 2, \dots, N\}$ denote a set of $N$ units and $2^{\mathcal{N}}$ its power set, i.e., the collection of all subsets of $\mathcal{N}$. Let $\boldsymbol\Pi = \{\Pi_1, \Pi_2, \dots, \Pi_K\}$ denote a subclassification (or partition) of these $N$ units into $K$ non-overlapping subclasses such that each subclass $\Pi_k$ consists of $|\Pi_k| \geq 2$ units, $\sum_{1\leq k\leq K} |\Pi_k| = N$, and their union $\bigcup_{1\leq k \leq K} \Pi_K$ recovers these $N$ units. The number of subclasses $K$ is not fixed a priori. Finally, let $\mathcal{A}$ be the set of all possible subclassifications. We first develop two notions of subclass homogeneity.
\begin{definition}[Average pairwise homogeneity]\rm
\label{def: nu homogeneity}
Let $\delta(i,j)$ denote a distance between unit $i$ and $j$. \emph{Average pairwise homogeneity} of subclass $\Pi_k$, denoted as $\nu(\Pi_k)$, refers to the following quantity:
\begin{equation}
    \label{eqn: avg dist among pi_k}
    \nu(\Pi_k) = \frac{1}{|\Pi_k| \times (|\Pi_k| - 1)} \sum_{i,j \in \Pi_k, i \neq j}\delta(i,j).
\end{equation}
\end{definition}
According to Definition \ref{def: nu homogeneity}, $\nu(\Pi_k)$ is the average distance of all pairwise comparisons among units in the subclass $\Pi_k$. For instance, in the TEE/CABG application with a hospital-preference-based instrumental variable exposure, $\delta(i, j)$ could measure some distance between IV-outcome confounders (e.g., patient composition and hospital characteristics) of hospital $i$ and $j$, and $\nu(\Pi_k)$ would then measure the homogeneity in these IV-outcome confounders of hospitals in the same subclass $\Pi_k$.

Associated with a subclassification $\boldsymbol\Pi$ and $\nu(\Pi_k)$ is the following homogeneity measure of $\boldsymbol\Pi$:
\begin{equation}
    \label{eqn: weighted avg dist of Pi}
    \nu(\boldsymbol\Pi; \mathcal{W}) = \sum_{1\leq k\leq K } w(\Pi_k) \times \nu(\Pi_k),
\end{equation}
where $\mathcal{W}$ is a shorthand for a pre-specified weighting scheme $w(\cdot): 2^{\mathcal{N}} \mapsto \mathbb{R}^{\geq 0}$ that maps each possible subclass $\Pi_k \in 2^{\mathcal{N}}$ to a non-negative real number.


\begin{definition}\rm
\label{def: optimal wrt nu}
A subclassification $\boldsymbol\Pi_{\text{opt}}^{\nu}$ is said to be optimal with respect to the homogeneity measure $\nu(\boldsymbol\Pi; \mathcal{W})$ if \[
\boldsymbol\Pi_{\text{opt}}^{\nu} = \argmin_{\boldsymbol\Pi \in \mathcal{A}}~\nu(\boldsymbol\Pi; \mathcal{W}).
\]
\end{definition}

In full matching with a binary exposure, each subclass consists of either one treated unit and multiple control units or one control unit and multiple treated units, and subclass homogeneity is measured by averaging over all pairwise comparisons between the treated unit and each control unit (or the control unit and each treated unit); see \citet[Section 3]{rosenbaum1991characterization}. This structure motivates a second sensible homogeneity measure as follows.

\begin{definition}[Star homogeneity]\rm
\label{def: star homo of subclass}
Let $i_k^\ast \in \Pi_k$ be a reference unit in subclass $\Pi_k$. \emph{Star homogeneity} refers to the following quantity:
\begin{equation}
    \label{eqn: avg star dist among pi_k}
    \nu_{\text{star}}(\Pi_k; i_k^\ast) = \frac{1}{|\Pi_k| - 1} \sum_{j \in \Pi_k, j \neq i_k^\ast} \delta(i_k^\ast, j).
\end{equation}
\end{definition}
Unlike $\nu(\Pi_k)$ which averages over all pairwise comparisons, $\nu_{\text{star}}(\Pi_k; i_k^\ast)$ first picks a reference unit (e.g., the unit with the highest or lowest dose in each subclass; see Section \ref{subsec: incorporate dose} for how to enforce this choice), compares all other units to this reference unit, and then averages over such comparisons. In the TEE/CABG application, if the hospital with the highest preference for TEE in the subclass $\Pi_k$ is chosen as the reference unit $i^\ast_k$, then $\nu_{\text{star}}(\Pi_k; i_k^\ast)$ measures how close in patient composition and hospital characteristics the other hospitals in the same subclass are compared to this highest-preference hospital.

Associated with a subclassification $\boldsymbol\Pi$, the star homogeneity, a vector of reference units $\mathbf{i}^\ast = (i_1^\ast, \dots, i_K^\ast)$, and a weighting scheme $\mathcal{W}$ is a second homogeneity measure of $\boldsymbol\Pi$:
\begin{equation}
    \label{eqn: weighted star dist of Pi}
    \nu_{\text{star}}(\boldsymbol\Pi; \mathbf{i}^\ast, \mathcal{W}) = \sum_{1\leq k\leq K} w(\Pi_k) \times \nu_{\text{star}}(\Pi_k; i_k^\ast).
\end{equation}

\begin{definition}\rm
\label{def: optimal wrt nu star}
A subclassification $\boldsymbol{\Pi}_{\text{opt}}^{\nu_{\text{star}}} = \{\Pi_{\text{opt}, 1}^{\nu_{\text{star}}}, \dots, \Pi_{\text{opt}, K}^{\nu_{\text{star}}}\}$ with reference units \\ $\mathbf{i}_{\text{opt}}^{\ast} = (i_{\text{opt}, 1}^{\ast}, \dots, i_{\text{opt}, K}^{\ast})$, $i^{\ast}_{\text{opt}, k} \in \Pi_{\text{opt}, k}^{\nu_{\text{star}}}$, is said to be optimal with respect to the homogeneity measure $\nu_{\text{star}}(\boldsymbol\Pi; \mathbf{i}^\ast, \mathcal{W})$ if 
\[
(\boldsymbol\Pi_{\text{opt}}^{\nu_{\text{star}}}, \mathbf{i}_{\text{opt}}^{\ast})  = \argmin_{\boldsymbol\Pi \in \mathcal{A}} \argmin_{i^\ast_k \in \Pi_k, 1 \leq k \leq K} \nu_{\text{star}}(\boldsymbol\Pi; \mathbf{i}^\ast, \mathcal{W}),
\]
where minimization is taken over all subclassifications \emph{and} all possible reference units in each subclass. 
\end{definition}

\begin{remark}\rm
In the special case of pair matching, it is easy to check that for all $\Pi_k \in \boldsymbol\Pi$ and for all $i_k^\ast \in \Pi_k$, $\nu_{\text{star}}(\Pi_k; i^\ast_k) = \nu(\Pi_k)$, i.e., two measures of subclass homogeneity reduce to the same measure. Moreover, under a weighting scheme that assigns the same weight to all matched pairs, we have $\boldsymbol{\Pi}_{\text{opt}}^{\nu_{\text{star}}} = \boldsymbol{\Pi}_{\text{opt}}^{\nu}$ and this optimal solution is precisely returned by an optimal non-bipartite pair matching algorithm (\citealp{lu2001matching,lu2011optimal}). 
\end{remark}

\section{Relationship between two optimal solutions}
\label{sec: relationship between two solutions}
A subclassification $\boldsymbol\Pi_{\text{opt}}^{\nu_{\text{star}}}$ is optimal with respect to the homogeneity measure $\nu_{\text{star}}(\cdot)$ and a weighting scheme $\mathcal{W}$. A natural question arises as to what can be said about its homogeneity under the other measure $\nu(\cdot)$, and how $\nu(\boldsymbol\Pi_{\text{opt}}^{\nu_{\text{star}}}; \mathcal{W})$ compares to the optimal $\nu(\cdot)$ homogeneity under the same weights. This section establishes a revealing relationship between $\nu(\boldsymbol\Pi_{\text{opt}}^{\nu_{\text{star}}}; \mathcal{W})$ and $\nu(\boldsymbol\Pi_{\text{opt}}^{\nu}; \mathcal{W})$.

\begin{lemma}\rm
\label{lemma: bound in nu star and nu}
Let $\Pi_k$ be a subclass with size $|\Pi_k|$ and $\delta(i, j)$ a distance that satisfies the triangle inequality. We have 
\[
\min_{i^\ast_k \in \Pi_k} \nu_{\text{star}}(\Pi_k; i^\ast_k) \leq \nu(\Pi_k),
\]
and 
\[
\nu(\Pi_k) \leq \frac{2(|\Pi_k| - 1)}{|\Pi_k|} \cdot \nu_{\text{star}}(\Pi_k; i^\ast_k),~\forall i_k^\ast \in \Pi_k.
\]
In particular, when $|\Pi_k| = 2$, we have $\nu_{\text{star}}(\Pi_k; i^\ast_k) = \nu(\Pi_k)$, $\forall i^\ast_k \in \Pi_k$.
\end{lemma}

\begin{proof}
All proofs in this article are in Supplementary Material B.
\end{proof}

Let $\nu^\ast_{\text{star}}(\Pi_k) := \min_{i^\ast_k \in \Pi_k} \nu_{\text{star}}(\Pi_k; i^\ast_k)$ be the minimum $\nu_{\text{star}}(\cdot)$ homogeneity of a subclass $\Pi_k$ among all reference units $i_k^\ast \in \Pi_k$. Define 
\[
\nu^\ast_{\text{star}}(\Pi; \mathcal{W}) = \sum_{1 \leq k \leq K} w(\Pi_k) \cdot \nu^\ast_{\text{star}}(\Pi_k).
\] 

\begin{corollary}\rm
\label{cor: bound nu between nu_start and 2 nu_star}
For any subclass $\Pi_k$, we have 
\[
\nu^\ast_{\text{star}}(\Pi_k) \leq \nu(\Pi_k) \leq \frac{2(|\Pi_k| - 1)}{|\Pi_k|} \cdot \nu^\ast_{\text{star}}(\Pi_k).
\] Moreover, for any subclassification $\boldsymbol\Pi$ and weighting scheme $\mathcal{W}$, we have
\[
\nu^\ast_{\text{star}}(\boldsymbol\Pi; \mathcal{W}) \leq \nu(\boldsymbol\Pi; \mathcal{W}) < 2 \nu^\ast_{\text{star}}(\boldsymbol\Pi; \mathcal{W}).
\]
\end{corollary}

Corollary \ref{cor: bound nu between nu_start and 2 nu_star} establishes a link between two homogeneity measures $\nu(\boldsymbol\Pi; \mathcal{W})$ and $\nu_{\text{star}}(\boldsymbol\Pi; \mathcal{W})$: any subclassification $\boldsymbol\Pi$ has its $\nu(\boldsymbol\Pi; \mathcal{W})$ sandwiched between $\nu^\ast_{\text{star}}(\boldsymbol\Pi; \mathcal{W})$ and $2\nu^\ast_{\text{star}}(\boldsymbol\Pi; \mathcal{W})$. Proposition \ref{prop: sandwich identity} is an important consequence of Corollary \ref{cor: bound nu between nu_start and 2 nu_star}.

\begin{proposition}\rm
\label{prop: sandwich identity}
Let $\boldsymbol\Pi_{\text{opt}}^{\nu_{\text{star}}}$ be an optimal partition with respect to the homogeneity measure $\nu_{\text{star}}(\cdot)$ and weighting scheme $\mathcal{W}$, and $\boldsymbol\Pi_{\text{opt}}^{\nu}$ optimal with respect to $\nu(\cdot)$ and the same weighting scheme. We have 
\[
\nu(\boldsymbol\Pi_{\text{opt}}^{\nu}; \mathcal{W}) \leq \nu(\boldsymbol\Pi_{\text{opt}}^{\nu_{\text{star}}}; \mathcal{W}) 
< 2\nu(\boldsymbol\Pi_{\text{opt}}^{\nu}; \mathcal{W}).
\]
In words, $\boldsymbol\Pi_{\text{opt}}^{\nu_{\text{star}}}$ is optimal under the homogeneity measure $\nu_{\text{star}}(\cdot)$, and its homogeneity under the other measure $\nu(\cdot)$ is no worse than the optimal homogeneity under $\nu(\cdot)$ by a factor of $2$.
\end{proposition}

In the computer science and operations research literature, an approximation algorithm refers to an algorithm that finds an approximate solution to an optimization problem with a provable guarantee on the distance between the approximate solution and the optimal solution; see \citet{vazirani2013approximation} and \citet{williamson2011design} for general discussion. A $\rho$-approximation algorithm refers to an approximation algorithm that returns an approximate solution $x^\ast_{\text{approx}}$ whose objective function value $f^\ast_{\text{approx}}$ is no worse than that of the optimal solution $f_{\text{opt}}$ by a factor of $\rho$, i.e.,
\[
f_{\text{opt}} \leq f^\ast_{\text{approx}} \leq \rho\times f_{\text{opt}}.
\]

Corollary \ref{cor: approx alg} is an immediate consequence of Proposition \ref{prop: sandwich identity}.

\begin{corollary}\rm
\label{cor: approx alg}
Let $\boldsymbol\Pi_{\text{opt}}^{\nu}$ and $\boldsymbol\Pi_{\text{opt}}^{\nu_{\text{star}}}$ be defined as in Definition \ref{def: optimal wrt nu} and Definition \ref{def: optimal wrt nu star} with respect to the same weighting scheme $\mathcal{W}$. If $\textsf{ALG}$ is an algorithm for finding $\boldsymbol\Pi_{\text{opt}}^{\nu_{\text{star}}}$, then $\textsf{ALG}$ is also a 2-approximation algorithm for finding $\boldsymbol\Pi_{\text{opt}}^{\nu}$.
\end{corollary}

Corollary \ref{cor: approx alg} is important and useful because efficient, polynomial-time algorithms exist for finding $\boldsymbol{\Pi}_{\text{opt}}^{\nu_{\text{star}}}$ with respect to suitable weights, as we demonstrate in the next section.

\section{An efficient, polynomial-time algorithm}
\label{sec: algorithm exist for Pi^astast}

\subsection{Graph, edge cover, and suitable weights}
We introduce some useful terminologies from the graph theory to carry forward the discussion. Let $G = (V, E)$ denote a graph with vertex set $V$ and edge set $E$. We use $e = (i, j),~i,j\in V$, to denote an edge connecting vertex $i$ and $j$, in which case we say vertex $i$ (and similarly $j$) is incident to edge $e = (i, j)$. A subset of edges $\mathcal{S} \subseteq E$ is said to form a star if $\mathcal{S} = \{(i, j_1), (i, j_2), \dots, (i, j_k)\}$; $i$ is often referred to as the internal node or center of the star, and $\{j_1, j_2, \dots, j_k\}$ leaves. 

An \emph{edge cover} of graph $G$ is a subset of edges $F \subseteq E$ such that all vertices in $G$ are incident to at least one edge in $F$. Let $\mathcal{F}$ denote the class of all edge covers of graph $G$, and each edge $e = (i, j)$ be associated with a nonnegative cost $c(i, j)$. The cost of an edge cover $F$ is defined to be
\begin{equation*}
    \text{COST}(F) = \sum_{(i, j) \in F} c(i, j).
\end{equation*}
Figure \ref{fig: illustrate costs of edge covers} gives two examples of an edge cover in the same graph. The cost of the edge cover in the left panel is $11.5$ and that in the right panel is $14$.
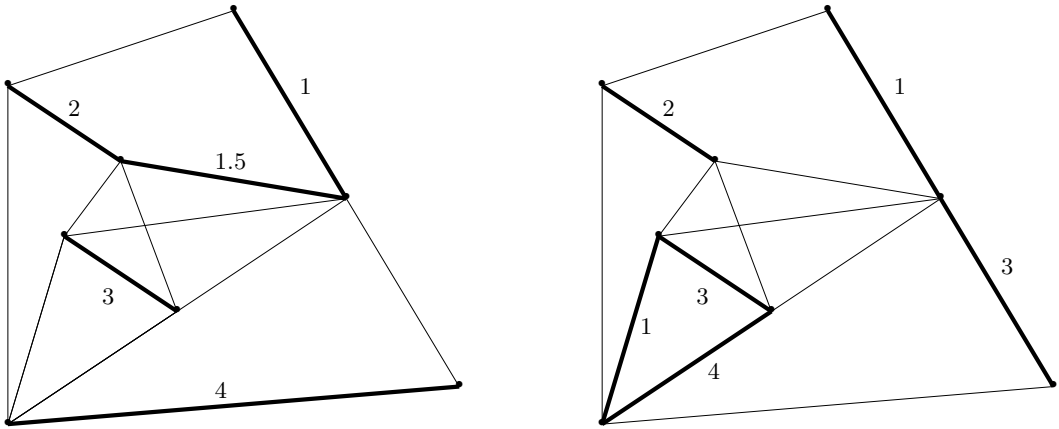
\begin{figure}[ht]
   \centering
     \subfloat{\begin{tikzpicture}[x = 1.5cm, y = 1cm]
  \foreach \Point in {(2.5, 2.5), (3, 5), (4, 4), (3.5, 6), (6.5, 3), (5.5, 5.5), (2.5, 7), (4.5, 8)}{
    \node at \Point {\textbullet};
    }
    
    \draw [ultra thin] (2.5, 2.5) -- (4, 4);
    \draw [ultra thin] (2.5, 2.5) -- (3, 5);
    \draw [ultra thin] (2.5, 2.5) -- (5.5, 5.5);
    \draw [ultra thin] (2.5, 2.5) -- (3, 5);
    \draw [ultra thin] (5.5, 5.5) -- (3, 5);
    \draw [ultra thick] (5.5, 5.5) -- (3.5, 6);
    \draw [ultra thick] (5.5, 5.5) -- (4.5, 8);
    \draw [ultra thick] (3, 5) -- (4, 4);
    \draw [ultra thin] (3.5, 6) -- (4, 4);
    \draw [ultra thin] (2.5, 2.5) -- (4, 4);
    \draw [ultra thin] (2.5, 2.5) -- (2.5, 7);
    \draw [ultra thick] (2.5, 7) -- (3.5, 6);
    \draw [ultra thin] (3, 5) -- (3.5, 6);
    \draw [ultra thin] (6.5, 3) -- (5.5, 5.5);
    \draw [ultra thin] (2.5, 7) -- (4.5, 8);
    \draw [ultra thick] (6.5, 3) -- (2.5, 2.5);

    \node [text width = 0.5 cm] at (5.25, 7) {$1$};
    \node [text width = 0.5 cm] at (4.5, 6) {$1.5$};
    \node [text width = 0.5 cm] at (3.2, 6.7) {$2$};
    \node [text width = 0.5 cm] at (3.5, 4.2) {$3$};
    \node [text width = 0.5 cm] at (4.5, 2.95) {$4$};
    \end{tikzpicture}}\hspace{1.5cm}
  \subfloat{\begin{tikzpicture}[x = 1.5cm, y = 1cm]
  \foreach \Point in {(2.5, 2.5), (3, 5), (4, 4), (3.5, 6), (6.5, 3), (5.5, 5.5), (2.5, 7), (4.5, 8)}{
    \node at \Point {\textbullet};
    }
    
    \draw [ultra thick] (2.5, 2.5) -- (4, 4);
    \draw [ultra thick] (2.5, 2.5) -- (3, 5);
    \draw [ultra thin] (2.5, 2.5) -- (5.5, 5.5);
    \draw [ultra thin] (2.5, 2.5) -- (3, 5);
    \draw [ultra thin] (5.5, 5.5) -- (3, 5);
    \draw [ultra thin] (5.5, 5.5) -- (3.5, 6);
    \draw [ultra thick] (5.5, 5.5) -- (4.5, 8);
    \draw [ultra thick] (3, 5) -- (4, 4);
    \draw [ultra thin] (3.5, 6) -- (4, 4);
    \draw [ultra thin] (2.5, 2.5) -- (4, 4);
    \draw [ultra thin] (2.5, 2.5) -- (2.5, 7);
    \draw [ultra thick] (2.5, 7) -- (3.5, 6);
    \draw [ultra thin] (3, 5) -- (3.5, 6);
    \draw [ultra thick] (6.5, 3) -- (5.5, 5.5);
    \draw [ultra thin] (2.5, 7) -- (4.5, 8);
    \draw [ultra thin] (6.5, 3) -- (2.5, 2.5);
    
    \node [text width = 0.5 cm] at (5.25, 7) {$1$};
    \node [text width = 0.5 cm] at (3.2, 6.7) {$2$};
    \node [text width = 0.5 cm] at (3.5, 4.2) {$3$};
    \node [text width = 0.5 cm] at (3, 3.8) {$1$};
    \node [text width = 0.5 cm] at (3.6, 3.2) {$4$};
    \node [text width = 0.5 cm] at (6.2, 4.6) {$3$};

    \end{tikzpicture}}
    \caption{\small Two edge covers (bold lines) of the same graph. The cost of the edge cover in the left panel is $1 + 2 + 1.5 + 3 + 4 = 11.5$, and the cost of the edge cover in the right panel is $1 + 3 + 2 + 3 + 1 + 4 = 14$.}
    \label{fig: illustrate costs of edge covers}
\end{figure}

Lemma \ref{lemma: nu star is cost of an edge cover} states that for a suitable choice of weights, homogeneity measure $\nu_{\text{star}}(\boldsymbol\Pi; \mathbf{i}^\ast, \mathcal{W})$ corresponds to the cost of a particular edge cover.

\begin{lemma}\rm
\label{lemma: nu star is cost of an edge cover}
Let $\boldsymbol\Pi = \{\Pi_1, \dots, \Pi_K\}$ be a partition, and $\mathcal{W}^{\text{suit}}$ a weighting scheme that assigns $w(\Pi_k) = |\Pi_k| - 1$ to subclass $\Pi_k$. Then
\begin{equation}
    \nu_{\text{star}}(\boldsymbol\Pi; \mathbf{i}^\ast, \mathcal{W}^{\text{suit}}) = \sum_{1\leq k\leq K} \sum_{j \in \Pi_k, j \neq i_k^\ast} \delta(i_k^\ast, j),
\end{equation}
and $ \nu_{\text{star}}(\boldsymbol\Pi; \mathbf{i}^\ast, \mathcal{W}^{\text{suit}})$ is equal to the cost of an edge cover with connected components $\{\Pi_1, \dots, \Pi_K\}$, each connected component $\Pi_k$ being a star with internal vertex $i^\ast_k$ and leaves $\{j \in \Pi_k, ~j\neq i^\ast_k\}$, and cost of any edge connecting two nodes $i$ and $j$ being $\delta(i, j)$.  
\end{lemma}

\subsection{A minimum cost edge cover induces an optimal subclassification with respect to suitable weights}
A minimum cost edge cover, i.e., the edge cover that attains the minimum cost among all edge covers of $G$, can be efficiently found in polynomial time (\citealp{schrijver2003combinatorial}); in fact, the problem of finding a minimum cost edge cover can be reduced to the problem of finding a minimum cost matching in an expanded non-bipartite graph. Moreover, Proposition \ref{prop: min cost edge cover yields an optimal partition} states that a minimum cost edge cover induces an optimal subclassification with respect to the homogeneity measure $\nu_{\text{star}}(\boldsymbol\Pi; \mathbf{i}^\ast, \mathcal{W}^{\text{suit}})$ when the edge cost is nonnegative.

\begin{proposition}\rm
\label{prop: min cost edge cover yields an optimal partition}
Let $G = (V, E)$ be a graph and $c: E\mapsto \mathbb{R}^{\geq 0}$ a nonnegative cost function. Then
\begin{enumerate}
    \item There exists a minimum cost edge cover whose connected components are all stars; call this minimum cost star-tiled edge cover $F_{\text{star}}^\ast$;
    \item Let the cost function $c(\cdot)$ of edge $e = (i, j)$ be the distance $\delta(i, j)$, then 
\[
\text{COST}(F_{\text{star}}^\ast) = \min_{\boldsymbol\Pi \in \mathcal{A}; ~i^\ast_k \in \Pi_k, 1 \leq k \leq K} \nu_{\text{star}}(\boldsymbol\Pi; \mathbf{i}^\ast, \mathcal{W}^{\text{suit}}).
\]
In other words, $F_{\text{star}}^\ast$ induces an optimal subclassification with respect to the homogeneity measure $\nu_{\text{star}}(\boldsymbol\Pi; \mathbf{i}^\ast, \mathcal{W}^{\text{suit}})$.
\end{enumerate}
\end{proposition}

\subsection{An efficient algorithm that finds minimum cost edge cover}
Algorithm \ref{alg: find min-cost edge cover} transforms the problem of finding a minimum cost edge cover into an optimal non-bipartite matching problem (\citealp{schrijver2003combinatorial}), the computation complexity of which is $O(|V|^3)$ in a graph with $|V|$ vertices. Algorithm \ref{alg: find min-cost edge cover} returns a minimum cost edge cover $F^\ast$; we may further process $F^\ast$ as described in the proof of Proposition \ref{prop: min cost edge cover yields an optimal partition} to obtain $F^\ast_{\text{star}}$, a minimum cost edge cover consisting of all stars. The complexity of finding a minimum cost edge cover is the same as optimal non-bipartite matching. The algorithm is further illustrated in Supplementary Material D. We will refer to the subclassification scheme induced by $F^\ast_{\text{star}}$ as a ``non-bipartite full match design."

\begin{algorithm} 
\SetAlgoLined
\caption{Finding a minimum cost edge cover for a graph $G = (V, E)$} \label{alg: find min-cost edge cover}
\vspace*{0.12 cm}
\KwIn{A graph $G = (V, E)$}
\vspace*{0.12 cm}
\begin{enumerate}
    \item Create a copy of $G = (V, E)$ with the same topology and edge cost; denote it as $G' = (V', E')$.
    \item For each $v \in V$ and its counterpart $v' \in V'$, add an edge $(v, v')$; a total of $|V|$ edges are added.
    \item Assign a cost equal to $2\mu(v)$ to each edge $(v, v')$, where $\mu(v)$ denotes the minimum cost among all edges incident to $v \in V$.
    \item Solve an optimal non-bipartite matching problem in the graph $G \cup G'$; let $M^\ast$ denote this optimal matching.
    \item Delete from $M^\ast$ any edge in $E'$; replace any edge of the form $(v, v')$ in $M^\ast$ with an edge $(v, u) \in E$ such that $\delta(v, u) = \mu(v)$; denote by $F^\ast$ the result;
    \item Return the minimum cost edge cover $F^\ast$.
\end{enumerate}
\end{algorithm}

\section{Additional design considerations}
\label{sec: other design aspects}
\subsection{Dose caliper}
\label{subsec: incorporate dose}
With a binary treatment, there is a distance between each treated unit and each control unit, and this distance unequivocally measures the closeness of the treated and control units in their observed covariates. With a continuous or many-level treatment/encouragement dose, homogeneity in observed covariates is still an important aspect; however, distances in this case may further take into account the treatment/encouragement dose in order to design matched sets that are homogeneous in observed covariates \emph{and} well-separated in their exposure doses (\citealp{lu2001matching}). 

This is in particular relevant in our TEE/CABG application with an IV-defined continuous exposure. It is widely acknowledged that confidence intervals obtained from weak instruments are often excessively long and non-informative (\citealp{imbens2005robust}). A large encouragement dose, on the other hand, would typically create stronger incentives for units to accept the treatment, increase the compliance rate, and eventually render the statistical inference substantially more powerful (\citealp{baiocchi2010building, heng2019instrumental, zhang2020bridging}). For instance, \citet{heng2019instrumental} derived the asymptotic relative efficiency (ARE) of some commonly-used test statistics when testing the same proportional treatment effect model (\citealp{small2008war}) with two instrumental variables of different strengths. They found that for a weaker IV with compliance rate $\iota_{1}$ to achieve the same efficiency as a stronger IV with compliance rate $\iota_{2}$ $(\iota_2 \geq \iota_1)$, the weaker IV needs to have a sample size $(\iota_2/\iota_1)^2$ times larger than that of a stronger IV. Analytic results of this kind provide incentives to separate exposure doses in the design stage.

How to pursue this design aspect in a non-bipartite full match? We borrow the idea of a ``caliper" from the literature on ``caliper matching" (\citealp{cochran1973controlling}) and ``propensity score caliper" (\citealp{rosenbaum1985constructing}). Let $\delta(i, j)$ measure the distance between observed covariates, and $Z_i$ and $Z_j$ the encouragement doses of unit $i$ and $j$, respectively. One straightforward way to incorporate the encouragement dose is to define a new distance $\delta'(i, j) = \delta(i, j) + C \times \mathbbm{1}\{|Z_i - Z_j| \leq \tau_0\}$, where $\tau_0$ is called a \emph{dose caliper} and $C$ a large penalty applied when $Z_i$ and $Z_j$ differ by less than or equal to the caliper size. Hence, a large $C$ would discourage an edge cover from connecting unit $i$ and $j$ whenever their doses are within the caliper size. Analogous to a propensity score caliper, a dose caliper may be implemented both as a hard constraint, by setting $C = \infty$ or equivalently removing edges $\{e = (i, j)~\text{such that}~|Z_i - Z_j| \leq \tau_0\}$, or as a soft constraint by setting $C$ to a large but finite number (\citealp[Section 2.3]{zhang2021match2C}). In Supplementary Material C.1, we illustrate the dose caliper using a simulation study.

A dose caliper facilitates interpreting the chosen reference unit in each subclass. Consider implementing a hard caliper with size $\tau_0$. Let $Z_{i^\ast_k}$ be the dose of the internal node of a subclass $\Pi_k$ and $\mathcal{Z}_{\text{leaves},k} = \{Z_j, j \in \Pi_k, j \neq i^\ast_k\}$ the collection of doses of leaves in $\Pi_k$. By definition of a hard dose caliper, the dose of the internal node $i^\ast_k$ necessarily satisfies \[
Z_{i^\ast_k} > Z_j + \tau_0 \quad \text{or} \quad Z_{i^\ast_k} \leq Z_j - \tau_0,
\]
for all $Z_j \in \mathcal{Z}_{\text{leaves},k}$. In the former case, the internal node corresponds to the unit with the highest dose (and at least $\tau_0$ higher in dose than any other unit in the same subclass) and we may view it as a \emph{pseudo-treated} unit in the subclass; in the latter case, it is one with the lowest dose and can be viewed as a \emph{pseudo-control} unit. This particular structure, one pseudo-treated and a variable number of pseudo-control units (or one pseudo-control and a variable number of pseudo-treated units) is analogous to that of full matching with a binary exposure (\citealp{hansen2004full,hansen2006optimal}); see Figure \ref{fig: illustrate dose caliper} for an illustration. In the TEE/CABG application, matching with a dose caliper forces each matched set to have one high-preference hospital and a few low-preference hospitals, or one low-preference hospital and a few high-preference hospitals, and comparisons in the health outcomes will be made among high-preference and low-preference hospitals within each formed subclass.

\begin{figure}[h]
   \centering
     \subfloat{\scalebox{0.8}{\begin{tikzpicture}
\draw[help lines, color=gray!30, dashed] (0,0) grid (7.9,4.9);
\draw[->,ultra thick] (0,0)--(8,0) node[right]{};
\draw[->,ultra thick] (0,0)--(0,5) node[above]{\large Exposure};


    \filldraw [black] (1.7,4) circle (2pt);
    \node [text width = 0.5 cm, text centered] at (1.7, 4.3) {$\tau_1$};
    
    \filldraw [black] (4,4) circle (2pt);
    \node [text width = 0.5 cm, text centered] at (4, 4.3) {$\tau_2$};

    \filldraw [black] (6,4) circle (2pt);
    \node [text width = 0.5 cm, text centered] at (6, 4.3) {$\tau_3$};
     \filldraw [black] (6.4,4) circle (2pt);
    \node [text width = 0.5 cm, text centered] at (6.4, 4.3) {$\tau_4$};
     \filldraw [black] (6.8,4) circle (2pt);
    \node [text width = 0.5 cm, text centered] at (6.8, 4.3) {$\tau_5$};

    \filldraw [black] (1,1) circle (2pt);
    \node [text width = 0.5 cm, text centered] at (1, 0.7) {$\gamma_1$};
    
    \filldraw [black] (2,1) circle (2pt);
    \node [text width = 0.5 cm, text centered] at (2, 0.7) {$\gamma_2$};

    \filldraw [black] (3.5,1) circle (2pt);
    \node [text width = 0.5 cm, text centered] at (3.5, 0.7) {$\gamma_3$};

    \filldraw [black] (3.8,1) circle (2pt);
    \node [text width = 0.5 cm, text centered] at (3.8, 0.7) {$\gamma_4$};

    \filldraw [black] (4.2,1) circle (2pt);
    \node [text width = 0.5 cm, text centered] at (4.2, 0.7) {$\gamma_5$};

    \filldraw [black] (4.5,1) circle (2pt);
    \node [text width = 0.5 cm, text centered] at (4.5, 0.7) {$\gamma_6$};

    \filldraw [black] (6.5,1) circle (2pt);
    \node [text width = 0.5 cm, text centered] at (6.5, 0.7) {$\gamma_7$};

    \draw [ultra thick, color = black] (1.7, 4) -- (2, 1); 
    \draw [ultra thick, color = black] (1.7, 4) -- (1, 1); 

    \draw [ultra thick] (4, 4) -- (4.2, 1);
    \draw [ultra thick] (4, 4) -- (4.5, 1); 
    \draw [ultra thick] (4, 4) -- (3.5, 1); 
    \draw [ultra thick] (4, 4) -- (3.8, 1);

    \draw [ultra thick] (6.5, 1) -- (6, 4); 
    \draw [ultra thick] (6.5, 1) -- (6.4, 4); 
    \draw [ultra thick] (6.5, 1) -- (6.8, 4);

    \node [text width = 0.5 cm, text centered] at (-1.3, 4) {\large Treated};
    \node [text width = 0.5 cm, text centered] at (-1.3, 1) {\large Control};
    \node [text width = 0.5 cm, text centered] at (4, -0.4) {\Large X};

\end{tikzpicture}}}
  \subfloat{\scalebox{0.8}{\begin{tikzpicture}
\draw[help lines, color=gray!30, dashed] (0,0) grid (7.9,4.9);
\draw[->,ultra thick] (0,0)--(8,0) node[right]{};
\draw[->,ultra thick] (0,0)--(0,5) node[above]{\large Dose};


    \filldraw [black] (1.7,4) circle (2pt);
    \node [text width = 0.5 cm, text centered] at (1.7, 4.3) {$i^\ast_1$};
    
    \filldraw [black] (1,1) circle (2pt);
    \filldraw [black] (1.5,1.2) circle (2pt);
    \filldraw [black] (2,1.5) circle (2pt);
    
    \draw [ultra thick, color = black] (1.7, 4) -- (2, 1.5); 
    \draw [ultra thick, color = black] (1.7, 4) -- (1.5, 1.2); 
    \draw [ultra thick, color = black] (1.7, 4) -- (1, 1); 
    

    \filldraw [black] (3.2,2.5) circle (2pt);
    \node [text width = 0.5 cm, text centered] at (3.2, 2.8) {$i^\ast_2$};
    
    \filldraw [black] (2.8,0.5) circle (2pt);
    \filldraw [black] (3.5,0.3) circle (2pt);
    
    \draw [ultra thick, color = black] (3.2, 2.5) -- (2.8, 0.5); 
    \draw [ultra thick, color = black] (3.2, 2.5) -- (3.5, 0.3); 
    

    \filldraw [black] (4.3,4.2) circle (2pt);
    \filldraw [black] (5,3.5) circle (2pt);

    \filldraw [black] (4.7, 1.8) circle (2pt);
    \node [text width = 0.5 cm, text centered] at (4.7, 1.5) {$i^\ast_3$};
    
    \draw [ultra thick, color = black] (4.7, 1.8) -- (4.3, 4.2); 
    \draw [ultra thick, color = black] (4.7, 1.8) -- (5, 3.5); 
  

    \filldraw [black] (6.5,3) circle (2pt);
    \node [text width = 0.5 cm, text centered] at (6.5, 3.3) {$i^\ast_4$};
    
    \filldraw [black] (6.2,0.2) circle (2pt);
    \filldraw [black] (6.5,0.5) circle (2pt);
    \filldraw [black] (7, 1) circle (2pt);
    
    \draw [ultra thick, color = black] (6.5, 3) -- (6.2, 0.2); 
    \draw [ultra thick, color = black] (6.5, 3) -- (6.5, 0.5); 
    \draw [ultra thick, color = black] (6.5, 3) -- (7, 1);

    \node [text width = 0.5 cm, text centered] at (-0.5, 0) {\large $0$};
    \node [text width = 0.5 cm, text centered] at (-0.5, 1) {\large $1$};
    \node [text width = 0.5 cm, text centered] at (-0.5, 2) {\large $2$};
    \node [text width = 0.5 cm, text centered] at (-0.5, 3) {\large $3$};
    \node [text width = 0.5 cm, text centered] at (-0.5, 4) {\large $4$};
    \node [text width = 0.5 cm, text centered] at (4, -0.4) {\Large X};
\end{tikzpicture}}}
    \caption{\small Left panel: full match with a binary treatment. Right panel: non-bipartite full match with a dose caliper $\tau_0 = 2$. From left to right are $\Pi_1$, $\Pi_2$, $\Pi_3$, and $\Pi_4$. The internal node in each subclass $\Pi_k$ is either the unit with the highest dose as in $\Pi_1$, $\Pi_2$, and $\Pi_4$, or the unit with the lowest dose as in $\Pi_3$. The internal node has dose at least $\tau_0 = 2$ larger than leaves in the same subclass.}
    \label{fig: illustrate dose caliper}
\end{figure}
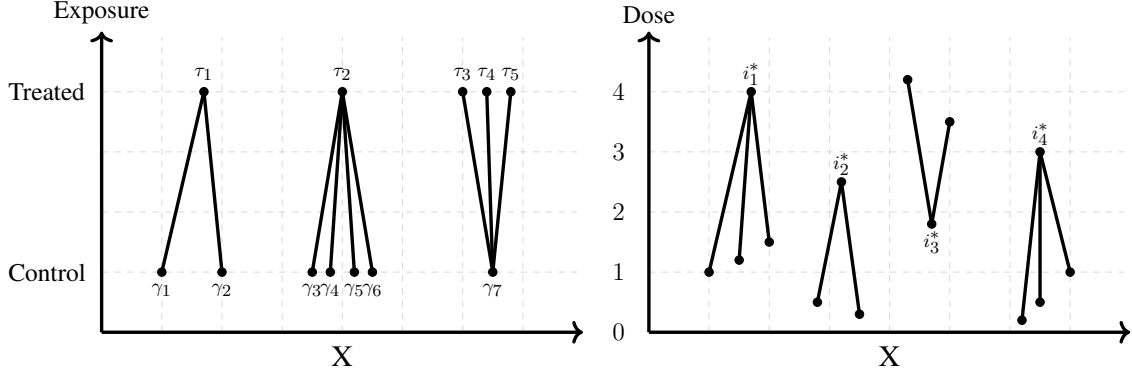

\subsection{Incorporating matched-sets-cardinality penalty}
\label{subsec: regulate matched sets number}
In a bipartite full matching, there are two parameters controlling for the maximum number of treated and control units in each matched set, respectively. Options \textsf{max.controls} and \textsf{min.controls} of function \textsf{fullmatch} in the \textsf{R} pacakge \textsf{optmatch} (\citealp{hansen2006optimal, hansen2007optmatch}) serve this purpose. For example, setting $\textsf{min.controls} = 0.25$ would restrict the matched set to have at most $4$ treated subjects for one control, and $\textsf{max.controls} = 4$ at most $4$ controls for one treated; together, they restrict the cardinality of matched sets to be at most $4 + 1 = 5$. 

In the context of non-bipartite matching with a continuous dose, we may also want to have some control over the size of subclasses and hence how many subclasses in a subclassification. To this end, we consider adding to a homogeneity measure a proper penalty on the cardinality of subclasses. Let the homogeneity measure be $\nu_{\text{star}}$, and consider the following modified homogeneity measure:
\begin{equation*}
    \nu^\lambda_{\text{star}}(\boldsymbol\Pi; \mathbf{i}^\ast, \mathcal{W}) = \nu_{\text{star}}(\boldsymbol\Pi; \mathbf{i}^\ast, \mathcal{W}) + \lambda \times \left\{\sum_{\Pi_k \in \boldsymbol\Pi} |\Pi_k| - 2 \right\}.
\end{equation*} 

\begin{definition}\rm
\label{def: optimal wrt nu star and lambda}
A subclassification $\boldsymbol{\Pi}_{\text{opt}, \lambda}^{\nu_{\text{star}}} = \{\Pi_{\text{opt},\lambda, 1}^{\nu_{\text{star}}}, \dots, \Pi_{\text{opt},\lambda, K}^{\nu_{\text{star}}}\}$ with reference units $\mathbf{i}_{\text{opt}, \lambda}^{\ast} = (i_{\text{opt},\lambda, 1}^{\ast}, \dots, i_{\text{opt}, \lambda, K}^{\ast})$, $i^{\ast}_{\text{opt}, \lambda, k} \in \Pi_{\text{opt}, \lambda, k}^{\nu_{\text{star}}}$, is said to be optimal with respect to the homogeneity measure $\nu^\lambda_{\text{star}}(\boldsymbol\Pi; \mathbf{i}^\ast, \mathcal{W})$ if 
\begin{equation*}
    \begin{split}
        (\boldsymbol\Pi_{\text{opt}, \lambda}^{\nu_{\text{star}}}, \mathbf{i}_{\text{opt}, \lambda}^{\ast})  &= \argmin_{\boldsymbol\Pi \in \mathcal{A}} \argmin_{i^\ast_k \in \Pi_k, 1 \leq k \leq K} \nu^\lambda_{\text{star}}(\boldsymbol\Pi; \mathbf{i}^\ast, \mathcal{W}) \\
        & = \argmin_{\boldsymbol\Pi \in \mathcal{A}} \argmin_{i^\ast_k \in \Pi_k, 1 \leq k \leq K} \left\{\nu_{\text{star}}(\boldsymbol\Pi; \mathbf{i}^\ast, \mathcal{W}) + \lambda \times \left\{\sum_{\Pi_k \in \boldsymbol\Pi} |\Pi_k| - 2 \right\}\right\}.
    \end{split}
\end{equation*}
\end{definition}
When $\lambda = 0$, this definition reduces to Definition \ref{def: optimal wrt nu star}; when $\lambda = \infty$, $\boldsymbol\Pi_{\text{opt}, \lambda = \infty}^{\nu_{\text{star}}}$ reduces to the solution to an optimal non-bipartite pair match because $|\Pi_k| - 2 = 0$ for matched pairs. As $\lambda$ increases from $0$ to $\infty$, we explore the middle ground between a subclassification that is optimal with respect to $\nu_{\text{star}}(\boldsymbol\Pi; \mathbf{i}^\ast, \mathcal{W})$ and an optimal non-bipartite pair matching.

With suitable weights $\mathcal{W}^{\text{suit}}$, we can find $\boldsymbol\Pi_{\text{opt}, \lambda}^{\nu_{\text{star}}}$ efficiently via a slightly modified version of Algorithm \ref{alg: find min-cost edge cover}; in fact, it suffices to modify Step 3 in Algorithm \ref{alg: find min-cost edge cover} as follows:
\begin{enumerate}
    \item[$3^\ast.$] Assign a cost equal to $2\mu(v) + 2\lambda$ to each edge $(v, v')$, where $\mu(v)$ denotes the minimum cost among all edges incident to $v \in V$. 
\end{enumerate}
Let $F^\ast_\lambda$ denote the output from the modified Algorithm \ref{alg: find min-cost edge cover}. Following a similar argument in the proof of Proposition \ref{prop: min cost edge cover yields an optimal partition}, we may further process $F^\ast_\lambda$ to obtain $F^\ast_{\text{star}, \lambda}$, an edge cover consisting of all stars. In Supplementary Material B.5, we prove a result analogous to Proposition \ref{prop: min cost edge cover yields an optimal partition}, which states that $F^\ast_{\text{star}, \lambda}$ induces a subclassification that is optimal with respect to $\nu^\lambda_{\text{star}}(\boldsymbol\Pi; \mathbf{i}^\ast, \mathcal{W}^{\text{suit}})$. Supplementary Material C.2 illustrates how to choose $\lambda$ with a simulation study.

\section{Simulation studies I: non-bipartite matching as a preprocessing step to remove bias in parametric causal inference with a continuous treatment dose}
\label{sec: nbp matching as preprocessing}
\subsection{Goal and structure}
\label{subsec: simulation studies I goal and structure}
It is widely acknowledged that with a binary treatment, combining statistical matching with regression adjustment renders analysis more robust to model misspecification and helps remove bias in treatment effect estimation (\citealp{rubin1973matching, rubin1979using}); hence, many authors advocate using statistical matching as a nonparametric preprocessing step before parametric causal inference (\citealp{ho2007matching, stuart2010matching}). The primary goal of this section is to assess if combining the non-bipartite full matching developed in this article and regression adjustment helps reduce model dependence and remove bias in the continuous treatment setting. 

Our simulation studies in this section can be compactly represented as a $2 \times 2 \times 2 \times 3 \times 2 \times 4$ factorial study with the following factors:

\begin{description}
\item \textbf{Factor 1:} treatment effect estimator: $\widehat{\beta}_{\text{reg}}$ and $\widehat{\beta}_{\text{reg, match}}$. 
\item \textbf{Factor 2:} dimension of covariates, $d$: $5$ and $10$.
\item \textbf{Factor 3:} sample size, $n$: $500$ and $2000$.
\item \textbf{Factor 4:} treatment dose model: a multi-level treatment $Z \sim \text{Uniform}\{-2, -1, 0, 1, 2\}$; two continuous treatments $Z \sim \text{Uniform}[1 - \sqrt{3}, 1 + \sqrt{3}]$ and $Z \sim \text{Exponential}(1)$ so that the continuous treatment Z has mean $1$ and variance $1$. 
\item \textbf{Factor 5:} observed covariates distribution: $\boldsymbol X \sim \text{Multivariate Normal}\left(\boldsymbol\mu, \boldsymbol \Sigma\right)$, with $\boldsymbol \mu = (cZ, 0, \dots, 0)^{\text{T}}$ and $\boldsymbol\Sigma = \bigl( \begin{smallmatrix} 4 & 0\\ 0 & \boldsymbol I_{d-1}\end{smallmatrix}\bigr)$ with $c = -2$ and $2$. 
\item \textbf{Factor 6:} response model: $Y \mid \boldsymbol X, Z \sim \text{Normal}\left(\mathbbm{1}\left\{\exp\{aX_1 + bX_2\} \leq 100\right\} + \beta Z + 1, 1\right)$ with $(a, b) = (-0.5, 0.5)$, $(0.5, -0.5)$, $(0.5, 0.5)$, and $(-0.5, -0.5)$, and $\beta = 1$.
\end{description}

Factor $1$ defines the procedures, and Factor $2$ through $6$ define the data generating processes. In particular, we considered three different models for a non-binary treatment, and closely followed \citet{rubin1979using} and \citet{zhang2021match2C} in specifying the data generating processes for the observed covariates $\boldsymbol X$ and the response surfaces $Y \mid \boldsymbol X, Z$; see \citet[Section 3]{rubin1979using} and \citet[Section 5.1]{zhang2021match2C} for some rationals behind these data generating processes. While both \citet{rubin1979using} and \citet{zhang2021match2C} considered an additive effect for a binary treatment, we considered an effect proportional to the magnitude of the treatment dose. Two treatment effect estimators being considered here are $\widehat{\beta}_{\text{reg}}$, the naive regression adjustment estimator, and $\widehat{\beta}_{\text{reg, match}}$, the regression adjustment estimator with a fixed effect for each matched set after non-bipartite full matching. We calculate the bias, standard error, and mean squared error of both estimators under each of the $96$ data generating processes defined by Factor 2 through 6.

\subsection{Simulation results}
Using non-bipartite matching as a preprocessing step followed by regression adjustment seems to help reduce bias and mean squared error in $92/96$ circumstances. Supplementary Material C.3 summarizes the bias and mean squared error of $\widehat{\beta}_{\text{reg}}$ and $\widehat{\beta}_{\text{reg, match}}$ under each data generating process. Figure \ref{fig: simulation I visualization} visualizes the gain in bias reduction under a wide range of data-generating processes. We also observe that for a fixed $d$, the gain from using statistical matching as a nonparametric preprocessing step seems to increase as $n$ increases, which is as expected because with a larger $n/d$ ratio, matched sets formed by non-bipartite matching tend to be more homogeneous; on the other hand, when the model is misspecified, a larger sample size does not seem to help remove bias of a naive regression estimator. For instance, when $d = 5$, $Z \sim \text{Exponential}(1)$, $(c, a, b) = (2, 0.5, -0.5)$, $\widehat{\beta}_{\text{reg}}$ has a bias of $1.515$ when $n = 500$ and a bias of $1.451$ when $n = 2000$, while $\widehat{\beta}_{\text{reg, match}}$ has a bias of $0.837$ when $n = 500$ and the bias reduces to $0.386$ when $n = 2000$.

Consistent with the binary treatment case studied in \citet{rubin1973matching, rubin1979using} and discussed in \citet{ho2007matching}, our simulation results seem to suggest that using non-bipartite matching as a nonparametric preprocessing step before regression analysis with a continuous treatment dose helps reduce model dependence and remove some bias.

\begin{figure}
    \centering
    \includegraphics[width=\textwidth]{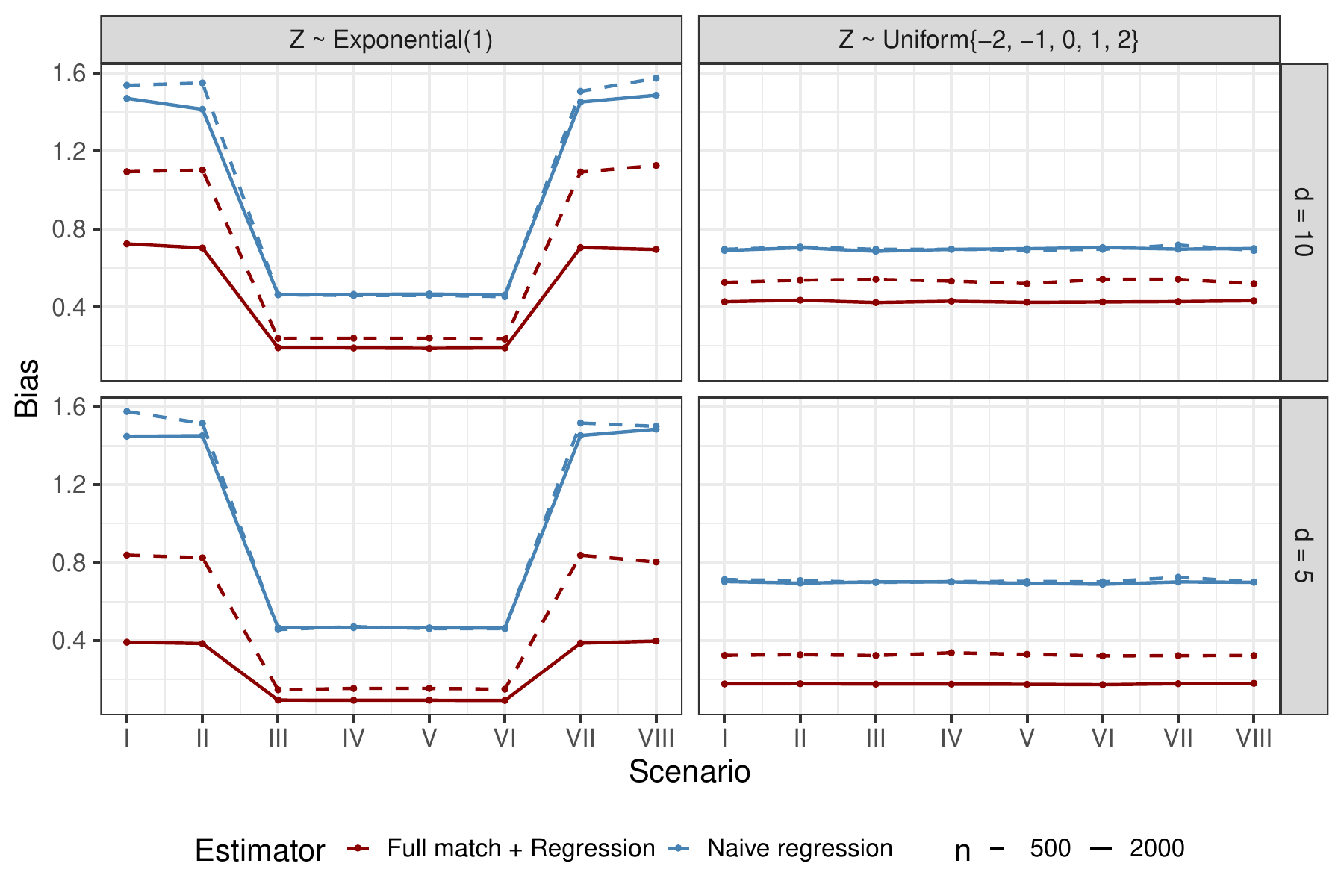}
    \caption{\small Simulation results: Bias of estimators $\widehat{\beta}_{\text{reg}}$ (blue) and $\widehat{\beta}_{\text{reg, match}}$ (red) for different data-generating processes. Scenarios I-VIII correspond to $(c, a, b) = (-2, -0.5, -0.5)$, $(-2, -0.5, 0.5)$, $(-2, 0.5, -0.5)$, $(-2, 0.5, 0.5)$, $(2, -0.5, -0.5)$, $(2, -0.5, 0.5)$, $(2, 0.5, -0.5)$, $(2, 0.5, 0.5)$.}
    \label{fig: simulation I visualization}
\end{figure}

\section{Simulation studies II: comparing non-bipartite full matching to pair matching}
\label{sec: nbp matching full vs pair}

\subsection{Goal and structure}
In this section, we systematically compare non-bipartite full matching with non-bipartite pair matching. We consider a continuous dose $Z \sim \text{Uniform}[0, 1]$, $d = 5$, and the following factors that define a data generating process:

\begin{description}
\item \textbf{Factor 1:} sample size, $n$: $500$ and $2000$.
\item \textbf{Factor 2:} observed covariates distribution: $\boldsymbol X \sim \text{Multivariate Normal}\left(\boldsymbol\mu, \boldsymbol \Sigma\right)$, with $\boldsymbol \mu = (cZ, 0, \dots, 0)^{\text{T}}$ and $\boldsymbol\Sigma = \bigl( \begin{smallmatrix} 2^2 & 0\\ 0 & \boldsymbol I_{d-1}\end{smallmatrix}\bigr)$ with $c = -2$, $-1$, $1$, and $2$. 
\end{description}
We compare the non-bipartite full matching procedure as in Algorithm \ref{alg: find min-cost edge cover} and the optimal non-bipartite pair matching procedure in \citet{lu2001matching,lu2011optimal}. This is the third factor:
\begin{description}
 \item \textbf{Factor 3:} matching procedure: non-bipartite full matching $\mathcal{M}_{\text{nbp, full}}$ and optimal non-bipartite pair matching $\mathcal{M}_{\text{nbp, pair}}$.
\end{description}
For both matching procedures, we consider the following distance:
\[
\delta(i, j) = \text{Mahalanobis distance}(i, j) + C \times \mathbbm{1}\{|Z_i - Z_j| \leq \tau_0\}.
\]
As discussed in Section \ref{subsec: incorporate dose}, $\delta(i, j)$ may incorporate the treatment/encouragement dose $Z$ by adjusting the dose caliper $\tau_0$ and letting $C$ be a large penalty. Throughout the simulations, we let $C = 100,000$ and $\tau_0$ be the fourth factor:
\begin{description}
 \item \textbf{Factor 4:} dose caliper size, $\tau_0$: $0$, $0.1$, $0.2$, $0.3$, and $0.4$.
\end{description}

To conclude, Factor $1$ and $2$ define the $2 \times 4 = 8$ data generating processes and Factor $3$ and $4$ define the $2 \times 5 = 10$ procedures to be studied.

\subsection{Measurements of success}
For a subclassification $\boldsymbol\Pi = \{\Pi_1, \dots, \Pi_K\}$, we compute $\nu(\Pi_k)$, the average Mahalanobis distance among all $(1/2)\times|\Pi_k| \times (|\Pi_k| - 1)$ pairwise comparisons in each subclass $\Pi_k$, and then report the $25$th, $50$th (median), $75$th, and $90$th empirical quantiles of $\{v(\Pi_k),~k = 1, \dots, K\}$. We also report two weighted averages of $\{v(\Pi_k),~k = 1, \dots, K\}$. The weighting scheme $\mathcal{W}^{\text{const}}$ assigns an equal weight to each matched set, regardless of its size, which corresponds to letting $w(\Pi_k) \propto 1$ in Definition \ref{def: optimal wrt nu}; denote by $\textsf{HM1}$ this first measure. The second weighting scheme $\mathcal{W}^{\text{suit}}$ assigns $w(\Pi_k) \propto |\Pi_k| - 1$ as described in Lemma \ref{lemma: nu star is cost of an edge cover}. Denote by $\textsf{HM2}$ this second measure.

Next, for each subclass, we compute $\nu^\ast_{\text{star}}(\Pi_k)$, the minimum $\nu_{\text{star}}(\Pi_k; i_k^\ast)$ (based on the Mahalanobis distance) among all $i_k^\ast \in \Pi_k$ as defined in Definition \ref{def: optimal wrt nu star}. We also report two weighted averages of $\{\nu^\ast_{\text{star}}(\Pi_k),~k = 1, \dots, K\}$: one with the weighting scheme $\mathcal{W}^{\text{const}}$ and the other $\mathcal{W}^{\text{suit}}$. Denote by $\textsf{HM3}$ and $\textsf{HM4}$ these two measures. Smaller values of $\textsf{HM1}$ through $\textsf{HM4}$ indicate better matched-sets homogeneity. Note that all four measures reduce to the same measure when the subclassification $\boldsymbol\Pi$ consists of only matched pairs.

We also consider a measurement of overall balance. In each subclass $\Pi_k$, let $\overline{\mathbf{X}}_{i, k, \text{high}}$ denote the average value of the ith observed covariate $\mathbf{X}_i$ of units with treatment dose greater than or equal to the median treatment dose, and $\overline{\mathbf{X}}_{i, k, \text{low}}$ that of units with treatment dose below the median. For instance, if the subclass consists of $5$ units, each with the first observed covariate $\mathbf{X}_i$ $\{1.5, 2, 1, 1.5, 2\}$ and treatment dose $\{0.1, 0.2, 0.3, 0.4, 0.5\}$, then $\overline{\mathbf{X}}_{i, k, \text{high}} = (1 + 1.5 + 2)/3 = 1.5$ and $\overline{\mathbf{X}}_{i, k, \text{low}} = (1.5 + 2)/2 = 1.75$ for this subclass. Let $d_{i} = \sum_{\Pi_k \in \boldsymbol\Pi} \overline{\mathbf{X}}_{i, k, \text{high}} - \sum_{\Pi_k \in \boldsymbol\Pi} \overline{\mathbf{X}}_{i, k, \text{low}}$ denote the difference in means of the $i$th covariate, and define $\textsf{SS} = \sum^d_{i = 1} d^2_i$ to be the sum of the squared differences over all $d = 5$ or $10$ observed covariates. In an ideal randomization experiment where treatment dose assignment is indeed randomized, distributions of observed covariates in the high and low dose groups are identical, and $\textsf{SS}$ is small. Hence, smaller $\textsf{SS}$ values signal better overall balance. 

For each subclass $\Pi_k$, we further calculate $\mu(\Pi_k)$, the average absolute ``internal-node-minus-leaf" difference in $Z$ as defined in Section \ref{subsec: incorporate dose}. We report the minimum, $25$th, $50$th (median), and $75$th empirical quantiles of $\{\mu(\Pi_k),~k = 1, \dots, K\}$. Finally, we report the number of matched set $K$, and the average pairwise Mahalanobis distance and balance measure $\textsf{SS}$ before matching.

\subsection{Simulation results}

Figure \ref{fig: simulation II} summarizes the simulation results for $9$ selected measures when $d = 5$, $n = 2000$, and $c = -2$. Simulation results for the other cases are qualitatively similar, and details can be found in Supplementary Material C.4.

We observe three consistent trends. First, when the dose caliper $\tau_0 = 0$ and the only goal of statistical matching is homogeneity in covariates $\mathbf{X}$, non-bipartite full matching and non-bipartite pair matching produce similar matched sets and have very similar performance with respect to all measures. Second, we observe that the number of matched sets in non-bipartite full matching decreases as $\tau_0$ increases. Third, for both full matching and pair matching, all four homogeneity measures \textsf{HM1} to \textsf{HM4} deteriorate as the dose caliper $\tau_0$ increases; however, non-bipartite full matching is capable of striking a better balance between homogeneity in covariates and heterogeneity in doses compared to non-bipartite pair matching. In fact, when $d = 5$, $n = 2000$, $c = -2$, and $\tau_0 \geq 0.3$, non-bipartite full matching outperforms non-bipartite pair matching \emph{simultaneously} in all $8$ measurements of matched-sets homogeneity ($4$ quantiles of $\nu(\Pi_k)$ and \textsf{HM1} through \textsf{HM4}), $4$ measurements of heterogeneity in treatment/encouragement doses ($4$ quantiles of $\mu(\Pi_k)$), \emph{and} the overall balance \textsf{SS}. For instance, when $\tau_0 = 0.4$, the median within-matched-sets Mahalanobis distance (i.e., $50$th of $\nu(\Pi_k)$) is equal to $0.898$ for pair matching and $0.756$ for full matching, and the overall balance measurement \textsf{SS} is $0.389$ for pair matching and as small as $0.105$ for full matching.

\begin{figure}
    \centering
    \includegraphics[width=1\textwidth]{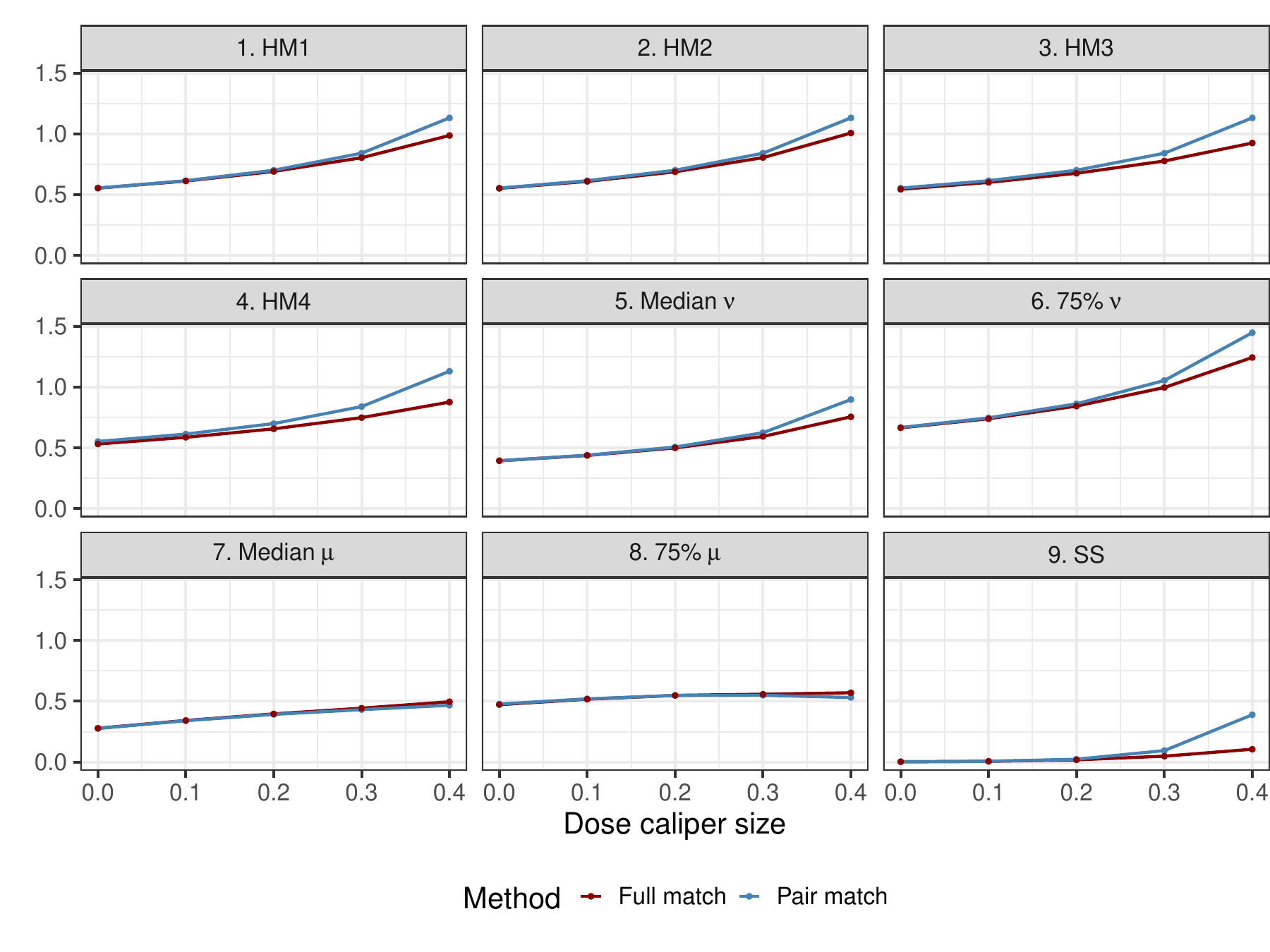}
    \caption{Simulation results: Comparing the non-bipartite pair match and non-bipartite full match when $d = 5$, $n = 2000$, and $c = -2$. The average pairwise Mahalanobis distance before matching is $10$ and $\textsf{SS}$ before matching is $1.010$. The number of matched pairs is each non-bipartite pair match is $1000$, and the average number of matched sets in a non-bipartite full match is $969$, $964$, $951$, $918$, and $842$ when the dose caliper size increases from $0$ to $0.4$.}
    \label{fig: simulation II}
\end{figure}

Finally, Figure \ref{fig: visualize} helps visualize the difference between an optimal non-bipartite pair match structure and an optimal non-bipartite full match structure using a small simulated dataset with $d = 3$, $n = 50$, and $c = -2$. To facilitate data visualization, we do a principle component analysis (PCA) and plot each unit's dose against its first principle component (PC1). Top left panel and top right panel depict the match structure of the optimal non-bipartite pair match and optimal non-bipartite full match, both with $\tau_0 = 0.3$. Two bottom panels eliminate matched pairs that are identical in two matches and focus on the match structure that are different in two matches. It is evident that the full match (corresponding to the bottom right panel) tends to connect units that are more different in the dose (i.e., larger difference in the y-axis) but similar in the first PC (i.e., small difference in the x-axis), compare to the pair match (corresponding to the bottom left panel).

\begin{figure}[H]
    \centering
    \includegraphics[width = 0.95\textwidth]{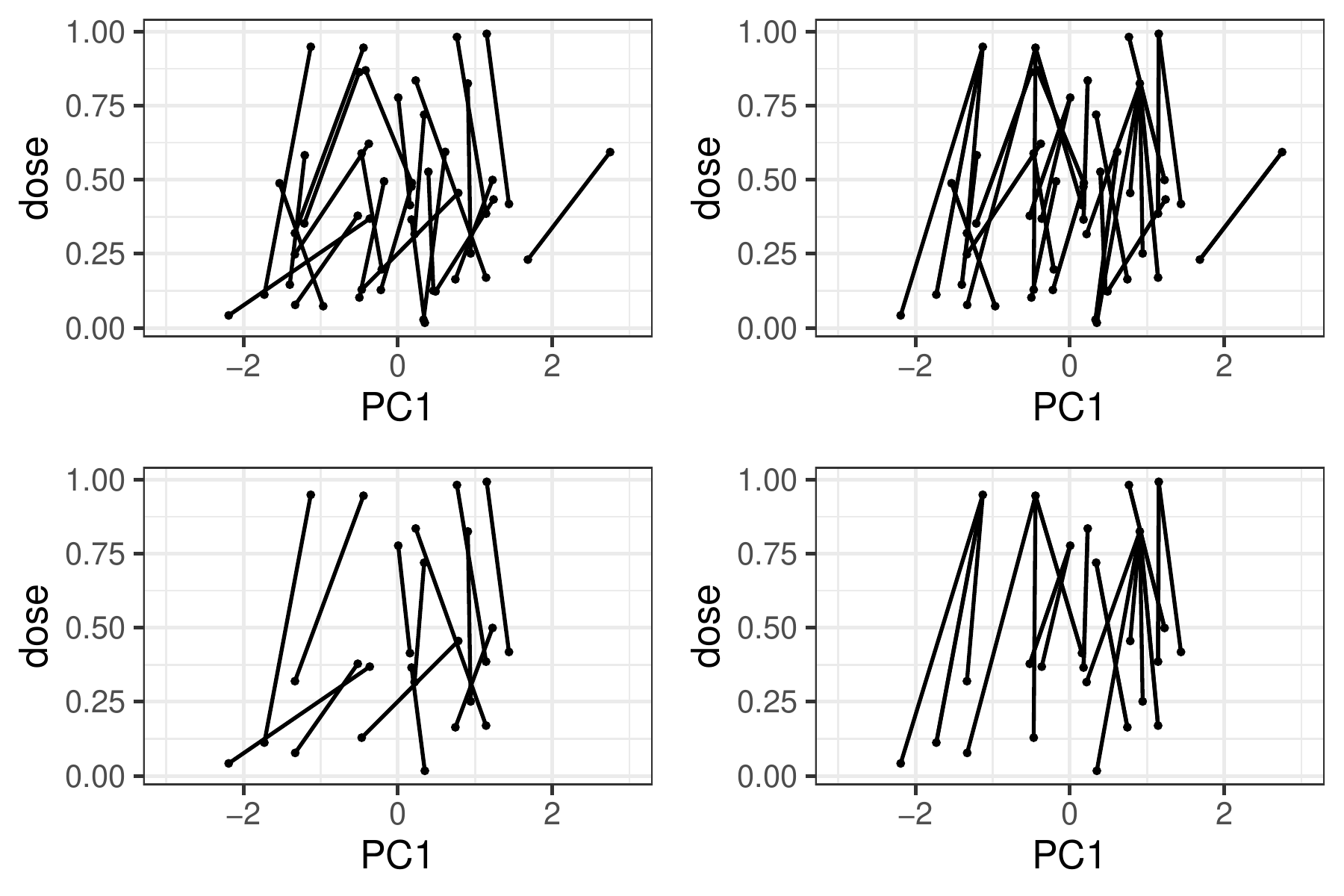}
    \caption{\small Visualizing the difference between a non-bipartite pair match and a non-bipartite full match. We generated a small dataset with $d = 3$, $n = 50$, and $c = -2$. Top left panel: optimal non-bipartite pair match with $\tau_0 = 0.3$. Top right panel: optimal non-bipartite full match with $\tau_0 = 0.3$. 
    Two bottom panels eliminate matched pairs that are identical in two matches and focus on the match structure that differ in two matches. In this simulated dataset, \textsf{HM1} = $1.46$ for pair match and $1.19$ for full match; average internal-node-minus-leaf difference in the dose is $0.48$ for pair match and $0.51$ for full match.}
    \label{fig: visualize}
\end{figure}

\section{The effect of TEE monitoring during CABG surgery on 30-day mortality}
\label{sec: real data study design}
\subsection{Data and study design}
We obtained data on patients undergoing isolated CABG surgery from Centers for Medicare and Medicaid Services (CMS). We identified patients' hospitals using the National Provider Identifier (NPI) numbers and obtained hospitals' characteristics data from the American Hospital Association Survey. Patient-level data were merged to hospitals' characteristics data using their unique NPI numbers. The study cohort consisted of all fee-for-service Medicare beneficiaries with a Part A (hospitalization) Medicare claim for isolated CABG surgery. Following \citet{MacKay2020_protocol}, we excluded (1) beneficiaries enrolled under managed care and not fee-for-service, (2) beneficiaries with less than six months of continuous enrollment in Medicare prior to the index admission for CABG surgery, (3) beneficiaries with age $< 65$ years, (4) beneficiaries without a cardiovascular or cardiac surgery-related Diagnosis Related Group (DRG) codes, (5) beneficiaries with a neurologic or stroke diagnosis as indicated by an ICD-9-cm code within the six months prior to the index admission or a stroke diagnosis with a “present on admission” (POA) indicator.

We follow \citet{MacKay2020_protocol,mackay2021association} and consider a cluster-level, instrumental variable analysis where each hospital defines a natural cluster and each hospital's preference for TEE usage (defined as the fraction of CABG surgeries using TEE monitoring) is considered a valid instrumental variable after controlling for patients' composition including average age, percentage of male patients, percentage of white patients, percentage of elective CABG surgeries, and percentage of patients having each of the following important comorbid conditions: arrhythmia, diabetes, congestive heart failure (CHF), hypertension, obesity, pulmonary diseases, and renal diseases, and hospital's characteristics including total hospitals beds, teaching status, presence of any cardiac intensive care unit, total number of full-time registered nurses, and annual cardiac surgical volume. Our goal in the design stage is to divide $1,217$ hospitals into subclasses with good subclass homogeneity, overall balance, and good separation in their encouragement doses. 

\subsection{Matched samples}
The first $4$ columns of Table \ref{tbl: real data balance table} summarize the patient composition and hospital characteristics of hospitals whose preference for using TEE during CABG surgery is above the median preference and those below the median preference. We observe a systematic difference between the ``above median" and ``below median" groups before matching: many standardized differences (defined as the difference in means divided by the standard deviation) are above $0.1$ and two-sample Kolmogorow--Smirnov tests suggest that the distribution of $6$ covariates, including annual cardiac surgical volume, hospital beds, etc, are statistically different at $0.01$ level.  

We then applied the developed non-bipartite full matching algorithm to the data using a dose-incorporating distance $\delta'(i, j) = \delta(i, j) + C \times \mathbbm{1}\{|Z_i - Z_j| \leq \tau_0\}$ with $\delta(i, j)$ being the Mahalanobis distance between $16$ observed pretreatment covariates, $C = 100,000$ a large penalty, $\lambda = 0$, and various choices of the dose caliper size $\tau_0$. We followed the advice in \citet{rubin2007design} and conducted the design without access to the outcome data in order to assure the objectivity of the design. 

In particular, non-bipartite full matching with $\tau_0 = 0.15$ divides these $1,217$ hospitals into $543$ matched pairs, $39$ matched sets of size $3$, $1$ of size $4$, and $2$ of size $5$. To get a sense of the balance after matching, we collect hospitals with higher doses in each matched set (including the one with median dose in a matched set with odd cardinality) and refer to them as the ``high dose" group. The ``high dose" group thus consists of $1\times543 + 2 \times 39 + 2 \times 1 + 3 \times 1 = 629$ hospitals. Similarly, we define the other hospitals as ``low dose" hospitals. In an ideal (yet unattainable) randomized controlled trial where the dose assignment within each matched set is indeed randomized, the ``high dose" and ``low dose" groups would have similar distributions of patient composition and hospital characteristics. Non-bipartite full matching seems to replicate this ideal experimental benchmark, as seen from the last $4$ columns in Table \ref{tbl: balance table}: the ``high dose" and ``low dose" groups have similar covariate distributions and in fact no Kolmogorov-Smirnov test is significant at $0.1$ level. Moreover, before matching, the median Mahalanobis distance among all $1,217$ hospitals is $14.14$, while the median ``average pairwise Mahalanobis distance" is as small as $1.54$ after matching. Matched sets also have a good separation in their encouragement doses: the average internal-node-minus-leaf difference in the encouragement dose is $0.46$ among all matched sets. In Supplementary Material E, we further report the covariate balance of non-bipartite full matches under other choices of $\tau_0$. We conduct inference with matched samples under $\tau_0 = 0.15$ because among all matches satisfying the stringent balance requirements (all standardized differences less than $0.1$ and no Kolmogorov-Smirnov tests significant at $0.05$ level), the match with $\tau_0 = 0.15$ produces the best separation in the encouragement doses. Compared to the original matched-pair design that discards approximately $20\%$ of hospitals, the full match design achieves similar balance while preserving all study units, and the outcome analysis based on the full match design is likely to have better generalizability (\citealp{cole2010generalizing}).

\begin{table}[ht]
\centering
\caption{  Covariate balance before and after non-bipartite full matching with a dose-incorporating Mahalanobis distance and $\tau_0 = 0.15$. $1,217$ hospitals are divided into $543$ matched pairs, $39$ matched sets of size $3$, $1$ of size $4$, and $2$ of size $5$. After matching, no two-sample Kolmogorov-Smirnov test comparing the covariate distributions in the ``high-dose" and ``low-dose" groups is significant at $0.05$ level.}
\label{tbl: real data balance table}
\resizebox{1.05\textwidth}{!}{
\begin{tabular}{lcccccccc}\hline
& \multicolumn{4}{c}{\multirow{2}{*}{\begin{tabular}{c}\Large{Before Matching}\end{tabular}}} & \multicolumn{4}{c}{\multirow{2}{*}{\begin{tabular}{c}\Large{After Matching}\end{tabular}}} \\ \\\cline{2-9}
& \multirow{4}{*}{\begin{tabular}{c}Below\\Median\\(n = 608)\end{tabular}} & \multirow{4}{*}{\begin{tabular}{c}Above\\Median\\(n = 609)\end{tabular}} & \multirow{4}{*}{\begin{tabular}{c}Std.\\ Diff.\end{tabular}} & \multirow{4}{*}{\begin{tabular}{c}K-S Test\\$P$-Value\end{tabular}} & \multirow{4}{*}{\begin{tabular}{c}Low\\ Dose\\(n = 588)\end{tabular}} & \multirow{4}{*}{\begin{tabular}{c}High\\Dose\\(n = 629)\end{tabular}} & \multirow{4}{*}{\begin{tabular}{c}Std.\\ Diff.\end{tabular}} & \multirow{4}{*}{\begin{tabular}{c}K-S Test\\$P$-Value\end{tabular}}\\ \\
\\ \\
\textsf{Patient Composition} &           &          &          \\ 
\hspace{0.5 cm}Mean age, yrs & \textbf{75.10} & \textbf{75.29} & -0.11 &\boldmath$<0.01$ & 75.18 & 75.21 & -0.02 &0.18\\ 
\hspace{0.5 cm}Male, \% & 0.67 & 0.69 & -0.15 &0.02 & 0.68 & 0.68 & -0.06 &0.92\\
\hspace{0.5 cm}White, \% & 0.85 & 0.85 & 0.01 & 0.05 & 0.85 & 0.85 & 0.01 & 0.20\\ 
\hspace{0.5 cm}Elective, \%  & 0.46 & 0.47 & -0.05 &0.32 & 0.47 & 0.47 & 0.00 &0.76\\ 
\hspace{0.5 cm}Diabetes, \%   & 0.17 & 0.17 & -0.02 &0.42 & 0.17 & 0.17 & 0.01 &0.93\\
\hspace{0.5 cm}Renal diseases, \%   & 0.09 & 0.09 & -0.07 &0.07 & 0.09 & 0.09 & -0.01 & 0.69\\
\hspace{0.5 cm}Arrhythmia, \%  & \textbf{0.11} & \textbf{0.12} & -0.10 & \boldmath$<0.01$ &0.11 & 0.11 & -0.02 & 0.08\\ 
\hspace{0.5 cm}CHF, \%  & 0.11 & 0.12 & -0.09 &0.12 & 0.12 & 0.12 & -0.01 &0.27\\ 
\hspace{0.5 cm}Hypertension, \%   & 0.29 & 0.30 & -0.05 &0.32 & 0.30 & 0.29 & 0.04 &0.96 \\ 
\hspace{0.5 cm}Obesity, \%  & 0.06 & 0.06 & -0.06 &0.03 & 0.06 & 0.06 & -0.03 &0.36\\ 
\hspace{0.5 cm}Pulmonary diseases, \% & \textbf{0.02} & \textbf{0.02} & -0.12 & \boldmath$<0.01$ &0.02 & 0.02 & -0.05 &0.43\\ 
\textsf{Hospital Characteristics} &           &          &          \\ 
\hspace{0.5 cm}Cardiac surgical volume & \textbf{456} & \textbf{571} & -0.21 &\boldmath$<0.001$ & 489 & 537 & -0.09 &0.09\\ 
\hspace{0.5 cm}Teaching hospital, yes/no & 0.15 & 0.20 & -0.13 &0.47 &0.18 & 0.18 & 0.01 &0.99\\ 
\hspace{0.5 cm}Hospital beds & \textbf{336} & \textbf{419} & -0.33 &\boldmath$<0.001$ & 370 & 386 & -0.06 &0.12 \\ 
\hspace{0.5 cm}Full-time registered nurses  & \textbf{534} & \textbf{722} & -0.33 &\boldmath $<0.001$ & 609 & 646 & -0.06 &0.14\\
\hspace{0.5 cm}Cardiac ICU, yes/no & 0.70 & 0.72 & -0.04 &0.99 & 0.71 & 0.71 & 0.02 &0.99\\  \hline
\end{tabular}}
\label{tbl: balance table}
\end{table}

\subsection{Statistical inference: notation, potential outcomes, and a cluster-level sharp null hypothesis}
Does using TEE during CABG surgery reduce patients' 30-day mortality? In this section, we generalize the cluster-level, non-bipartite pair match set-up considered in \citet{zhang2020bridging} to the current full match setting, and discuss how to test Fisher's sharp null hypothesis of no treatment effect under the new design.
 
Suppose we have formed $K$ matched sets, indexed by $k = 1, \dots, K$, each with $n_k$ hospitals, indexed by $j = 1, \dots, n_k$, so that index $kj$ uniquely identifies one hospital and there are a total of $N = \sum_{k = 1}^K n_k$ hospitals in total. Each hospital is associated with hospital-level covariates $\mathbf{x}_{kj}$ and a hospital-level continuous instrumental variable (or encouragement dose) $Z^{\text{obs}}_{kj}$. There are $N_{kj}$ patients in each hospital $kj$, indexed by $i = 1, \dots, N_{kj}$, so that index $kji$ uniquely identifies one patient. Each patient is associated with a treatment indicator $D^{\text{obs}}_{kji}$, outcome of interest $R^{\text{obs}}_{kji}$, and individual-level covariates $\mathbf{x}_{kji}$. In our application, we have formed $585$ matched sets so $K = 585$; $n_k$ is the number of hospitals in each matched set so $n_k = 2$, $3$, $4$, or $5$ in our design; hospital-level instrumental variable $Z_{kj} \in [0, 1]$ is hospital's preference for TEE during CABG surgery; $N_{kj}$ is the number of patients undergoing CABG surgery in hospital $kj$; $D_{kji}$ is a binary indicator equal to $1$ if patient $kji$ receives TEE monitoring and $0$ otherwise; $R_{kji}$ is patient $kji$'s 30-day mortality status; finally, $\mathbf{x}_{kj}$ describes hospital $kj$'s characteristics and $\mathbf{x}_{kji}$ patient $kji$'s characteristics. Following \citet{zhang2020bridging}, we assume that after controlling for patient composition and hospital characteristics, preference for TEE usage is a valid cluster-level instrumental variable.

Let $D_{kji}(Z_{kj} = z_{kj})$ denote the potential treatment received of patient $kji$ when the hospital-level IV $Z_{kj}$ is set to $z_{kj}$, and $ \mathbf{D}_{kj}(Z_{kj})$ is a shorthand for $\big(D_{kj1}(Z_{kj}), \dots, D_{kjN_{kj}}(Z_{kj})\big)$. Let $R_{kji}\big(Z_{kj} = z_{kj}, \mathbf{D}_{kj}(Z_{kj}) = \mathbf{d}_{kj}\big)$ denote unit $kji$'s potential outcome under $Z_{kj} = z_{kj}$ and $\mathbf{D}_{kj}(Z_{kj}) = \mathbf{d}_{kj}$. Under exclusion restriction, we have $R_{kji}\big(Z_{kj}, \mathbf{D}_{kj}(Z_{kj})\big) = R_{kji}\big(\mathbf{D}_{kj}(Z_{kj})\big)$. Finally, let $\mathcal{Z}^{\text{obs}}_{k} = \{Z^{\text{obs}}_{k1}, \dots, Z^{\text{obs}}_{kn_{kj}}\}$ denote the collection of IV doses in matched set $k$.

A cluster-level Fisher's sharp null hypothesis states that 
\begin{equation}
\begin{split}
     H_{0, \text{sharp}}: \quad &N^{-1}_{kj}\sum_{i = 1}^{N_{kj}} R_{kji}\big(\mathbf{D}_{kj}(Z_{kj} = z)\big) - N^{-1}_{kj}\sum_{i = 1}^{N_{kj}} R_{kji}\big(\mathbf{D}_{kj}(Z_{kj} = z')\big) \\
     &\qquad=\beta \left(N^{-1}_{kj}\sum_{i = 1}^{N_{kj}} D_{kji}(Z_{kj} = z) - N^{-1}_{kj}\sum_{i = 1}^{N_{kj}} D_{kji}(Z_{kj} = z')\right), 
\end{split}
\end{equation}
for all $k$, $j$, and $z, z' \in \mathcal{Z}^{\text{obs}}_k$ such that $z' < z$. The null hypothesis $H_{0, \text{sharp}}$ generalizes the proportional treatment effect model in \citet{small2008war} and \citet{zhang2020bridging}, and states that the mean difference of the hospital-aggregate potential outcomes when comparing any pair of two IV doses $z, z' \in \mathcal{Z}^{\text{obs}}_k$ is proportional to the mean difference of the potential treatment received under IV dose $z$ and $z'$ with structural parameter $\beta$; see \citet{baiocchi2010building} and \citet{zhang2020bridging} for other causal null hypotheses that may be of interest. We consider testing the causal null hypothesis $\beta = 0$, i.e., TEE received during CABG has no effect whatsoever on patients' 30-day mortality, against $\beta < 0$, i.e., TEE received during CABG lowers patients' 30-day mortality.

\subsection{Randomization-based inference}
There are three key ingredients to perform a Fisher-style randomization-based test: a sharp null hypothesis, a randomized treatment assignment scheme, and a test statistic (\citealp{rosenbaum2002observational, rosenbaum2010design, Ding2016variation}).

\begin{enumerate}
    \item \textbf{Sharp null hypothesis:} $H_{0, \text{sharp}}$ is sharp null hypothesis that allows us to impute the potential aggregate-outcome of cluster $kj$ under any IV doses $z \in \mathcal{Z}^{\text{obs}}_{k}$; see, e.g., \citet[Section 3.1]{zhang2020bridging}.
    \item \textbf{Randomized treatment dose assignment:} 
In a typical matched-pair design with $I$ matched pairs, there are a total of $2^I$ possible randomization configurations. In a full match design, within each matched set of $n_k$ hospitals, there are $n_k!$ many IV dose assignments, each with equal probability; therefore, there are a total of $\prod^K_{k = 1} n_k!$ randomizations induced by a full match design. Let $\mathcal{Z}$ denote this collection of all randomizations and $\mathbf{z} \in \mathcal{Z}$ one realization.

\item \textbf{Test statistic:} 
In principle, any test statistic $t(\mathbf{Z}, \mathbf{R}(\mathbf{Z}))$ that depends on the treatment dose assignment $\mathbf{Z}$ and the potential outcomes only via potential outcomes' dependence on $\mathbf{Z}$ can be combined with the randomization scheme to deliver a valid test for $H_{0, \text{sharp}}$; see, e.g., \citet{Ding2016variation}. With a binary treatment, a commonly-used test statistic for a full match design is the rank-sum test; see \citet{rosenbaum2002observational, rosenbaum2004design, heng2019increasing}. We modify the rank-sum test statistic to reflect the continuous dose. Let $\mathbf{Z} = \{Z^{\text{obs}}_{11}, \dots, Z^{\text{obs}}_{Kn_k}\}$, $R^{\text{obs}}_{kj\bigcdot} = N^{-1}_{kj} \left\{\sum_{i = 1}^{N_{kj}} R^{\text{obs}}_{kji}\right\}$, and $\mathbf{R} = \{R^{\text{obs}}_{11\bigcdot}, \dots, R^{\text{obs}}_{Kn_k\bigcdot}\}$. Consider the following double rank sum statistic:
\begin{equation}
\label{eqn: double rank test statistic}
    T_{\text{double rank}} = \frac{1}{N^2} \sum_{k = 1}^K \sum_{j = 1}^{n_k} q_1(Z^{\text{obs}}_{kj} \mid \mathbf{Z})\times q_2(R^{\text{obs}}_{kj\bigcdot} \mid \mathbf{R}),
\end{equation}
where $q_1(Z^{\text{obs}}_{kj} \mid \mathbf{Z})$ is the rank of $Z^{\text{obs}}_{kj}$ among all doses $\mathbf{Z}$, and $q_2(R^{\text{obs}}_{kj\bigcdot})$ the rank of $R^{\text{obs}}_{kj\bigcdot}$ among all the responses.
\end{enumerate}

Researchers first impute all missing potential outcomes under $H_{0, \text{sharp}}$ and then enumerate all $|\mathcal{Z}| = \prod^K_{k = 1} n_k!$ possible dose assignments. For each enumerated $\mathbf{Z}' = \mathbf{z} \in \mathcal{Z}$, calculate the corresponding $\mathbf{R}'$ under $\mathbf{Z}'$ and the test statistic $T'_{\text{double rank}}$. The distribution of $T'_{\text{double rank}}$ is then the exact null distribution of the test statistic $T_{\text{double rank}}$ under $H_{0, \text{sharp}}$ and conditional on the matched samples. By comparing $T_{\text{double rank}}$ to this exact null distribution, the exact $p$-value is obtained. In practice, researchers may sample with replacement from $\mathcal{Z}$ and report a Monte Carlo $p$-value.

\subsection{Results}
For $585$ matched sets we formed in the design stage, we generated the reference distribution using $1,000,000$ samples from all possible $2^{543} \times 3^{39} \times 4^1 \times 5^2$ randomizations; see Figure \ref{fig: real data reference dist}. We calculated $T_{\text{double rank}} = 294.27$; hence, one-sided $p$-value is $0.020$ and the null hypothesis $H_{0, \text{sharp}}$ is rejected at $0.05$ level in favor of the alternative hypothesis that $\beta < 0$, i.e., using TEE during CABG surgery lowers patients' 30-day mortality rate.

\begin{figure}[h]
    \centering
    \includegraphics[width = 0.8\columnwidth]{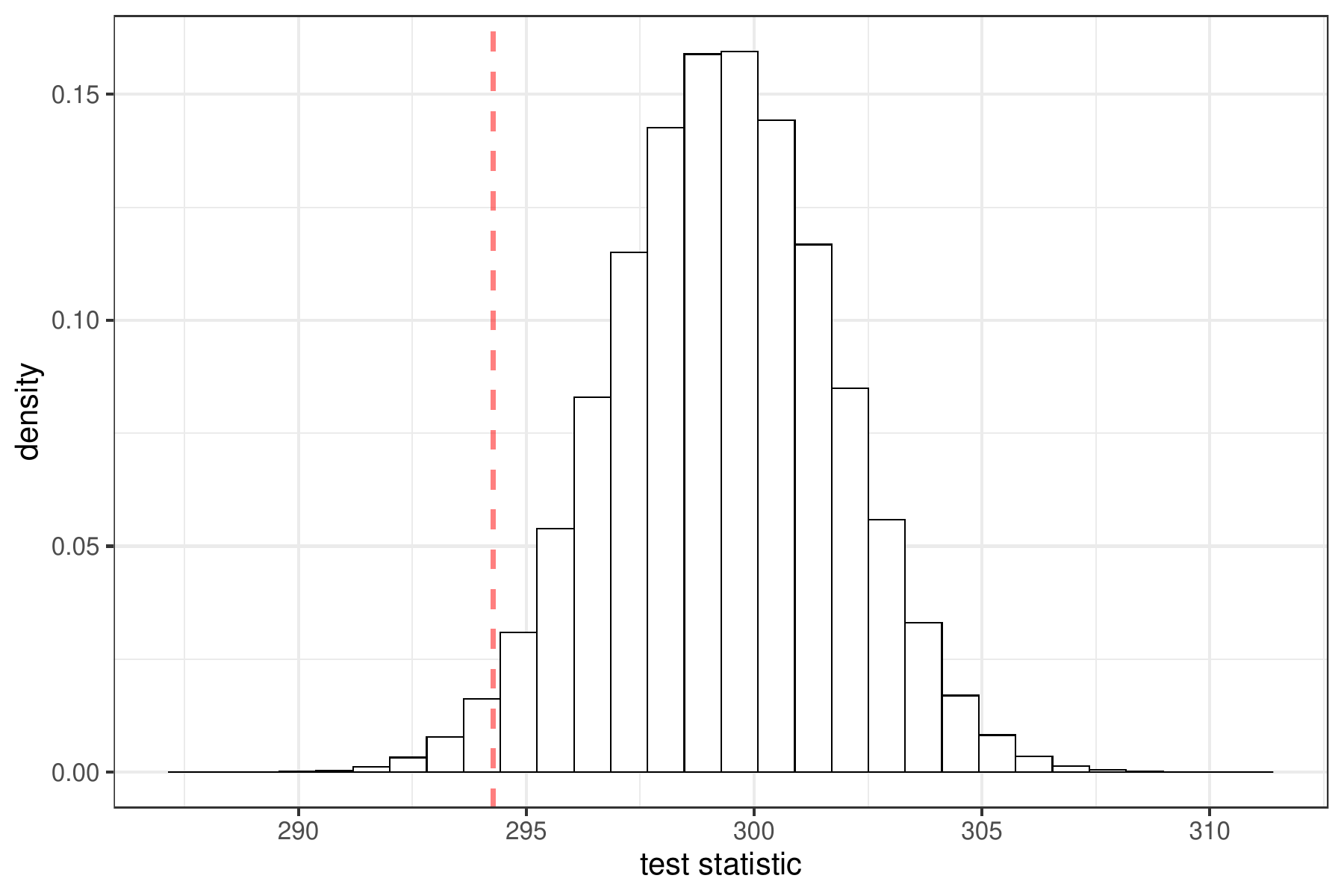}
    \caption{\small Reference distribution and the test statistic evaluated at the observed data. Exact one-sided $p$-value is $0.020$.}
    \label{fig: real data reference dist}
\end{figure}

\section{Discussion}
\label{sec: discussion}
In this paper we have systematically studied statistical matching and subclassification with a many-level or continuous exposure dose. We propose two optimality criteria for subclassification, each based on a natural subclass homogeneity measure. We characterize the relationship between these two criteria and leverage this relationship to develop an efficient polynomial-time algorithm that finds a subclassification that is guaranteed to be optimal with respect to one criterion and near-optimal with respect to the other criterion. 

Our developed algorithm also allows empirical users to control separation in the exposure dose and cardinalities of formed subclasses. There are three tuning parameters involved in our flexible algorithm: dose caliper size $\tau_0$, dose caliper penalty $C$, and cardinality penalty $\lambda$. In many practical situations, we recommend setting $\tau_0$ to a minimum dose difference that would yield a meaningful difference in the potential outcomes, and $C = \infty$ (or a very large number) to enforce dose separation specified by $\tau_0$. This way to specify tuning parameters $(\tau_0, C)$ is similar to setting propensity score calipers in a bipartite match; see, e.g., \citet[Section 2.3]{zhang2021match2C}. We recommend setting the cardinality penalty $\lambda = 0$ by default to deliver an optimal full match; in the case where some subclasses have too many study units, $\lambda$ should be gradually increased.

Our extensive simulations suggest that non-bipartite matching combined with regression adjustment helps remove bias in parametric causal inference; thus, we would recommend routinely using non-bipartite matching as a pre-processing step, as advocated by many researchers (\citealp{rubin1973matching, rubin1979using, ho2007matching, stuart2010matching}) in a binary treatment setting. Moreover, we found non-bipartite full match is advantageous over non-bipartite pair match in separating the treatment/encouragement doses \emph{and} maintaining good subclass homogeneity and overall balance; therefore, the new design may be particularly useful in instrumental variable studies where separation of the IVs (or encouragement doses) would render outcome analysis much more efficient (\citealp{baiocchi2010building}).


\begin{supplement}
\stitle{Proofs, additional simulation studies, and more details on the application}
\sdescription{Supplementary Material A contains a detailed literature review on bipartite and non-bipartite matching. Supplementary Material B.1-B.5 contain proofs of Lemma \ref{lemma: bound in nu star and nu}, Lemma \ref{lemma: nu star is cost of an edge cover}, Corollary \ref{cor: bound nu between nu_start and 2 nu_star}, Proposition \ref{prop: sandwich identity} and \ref{prop: min cost edge cover yields an optimal partition}. Supplementary Material B.6 proves that the output from the modified Algorithm 1 $F^\ast_{\text{star}, \lambda}$ induces a subclassification that is optimal with respect to $\nu^\lambda_{\text{star}}(\boldsymbol\Pi; \mathbf{i}^\ast, \mathcal{W}^{\text{suit}})$. Supplementary Material C.1 and C.2 illustrate a dose caliper and choice of $\lambda$ using simulation studies. Supplementary Materials C.3 and C.4 provide additional simulation results. Supplementary Material D illustrates how to find a minimum-cost edge cover. Supplementary Material E provides further details on statistical matching in the application with different choices of the tuning parameter $\tau_0$.}
\end{supplement}
\begin{supplement}
\stitle{nbpfull\_0.1.0.zip}
\sdescription{\textsf{R} code implementing the proposed non-bipartite full match algorithm.}
\end{supplement}


\bibliographystyle{imsart-nameyear} 
\bibliography{paper-ref}       

\end{document}